\documentclass{article}
\usepackage{PRIMEarxiv}
\usepackage{Definitions}

\usepackage[utf8]{inputenc} 
\usepackage[T1]{fontenc}    
\usepackage{hyperref}       
\usepackage{booktabs}       
\usepackage{amsfonts}       
\usepackage{nicefrac}       
\usepackage{microtype}      
\usepackage{lipsum}
 \usepackage{fancyhdr}       
\usepackage{graphicx}       

\usepackage[linesnumbered, ruled, vlined]{algorithm2e}
\usepackage{multirow,multicol}

\newcommand{\tikzscale}{{0.7}}

\usepackage{tikz,pgfplots,pgfplotstable}
\usetikzlibrary{positioning,fit,shapes}
\usetikzlibrary{decorations.pathreplacing}
\usetikzlibrary{decorations.markings}
\usetikzlibrary{arrows}

\pgfdeclarelayer{background}
\pgfdeclarelayer{foreground}
\pgfsetlayers{background,main,foreground}

\makeatletter
\tikzset{multicircle/.style  args={#1, #2}{%
 alias=tmp@name, %
  postaction={%
    insert path={
     \pgfextra{%
     \pgfpointdiff{\pgfpointanchor{\pgf@node@name}{center}}%
                  {\pgfpointanchor{\pgf@node@name}{east}}%
     \pgfmathsetmacro\insiderad{\pgf@x}%
        \fill[white] (\pgf@node@name.center)  circle (\insiderad-\pgflinewidth);%
        \draw[#2] (\pgf@node@name.center)  circle (\insiderad-\pgflinewidth);%
        \fill[#2] (\pgf@node@name.center)  -- ++(0:\insiderad-\pgflinewidth) arc (0:#1:\insiderad-\pgflinewidth)--cycle;%
        }}}}}
\makeatother

\definecolor{yafaxiscolor}{rgb}{0.3, 0.3, 0.3}
\definecolor{yafcolor1}{rgb}{0.4, 0.165, 0.553}
\definecolor{yafcolor2}{rgb}{0.949, 0.482, 0.216}
\definecolor{yafcolor3}{rgb}{0.47, 0.549, 0.306}
\definecolor{yafcolor4}{rgb}{0.925, 0.165, 0.224}
\definecolor{yafcolor5}{rgb}{0.141, 0.345, 0.643}
\definecolor{yafcolor6}{rgb}{0.965, 0.933, 0.267}
\definecolor{yafcolor7}{rgb}{0.627, 0.118, 0.165}
\definecolor{yafcolor8}{rgb}{0.878, 0.475, 0.686}
\definecolor{yafcolor9}{rgb}{0.965, 0.733, 0.767}

\newlength{\yafaxispad}
\newlength{\yaftlpad}
\newlength{\yaflabelpad}
\newlength{\yafaxiswidth}
\newlength{\yafticklen}
\makeatletter
\def\pgfplots@drawtickgridlines@INSTALLCLIP@onorientedsurf#1{}
\makeatother

\newcommand{\yafdrawxaxis}[2]{
  \pgfplotstransformcoordinatex{#1}\let\xmincoord=\pgfmathresult 
  \pgfplotstransformcoordinatex{#2}\let\xmaxcoord=\pgfmathresult 
  \pgfsetlinewidth{\yafaxiswidth} 
  \pgfsetcolor{yafaxiscolor}
  \pgfpathmoveto{\pgfpointadd{\pgfpointadd{\pgfplotspointrelaxisxy{0}{0}}{\pgfqpointxy{\xmincoord}{0}}}{\pgfqpoint{-0.5\yafaxiswidth}{\yafaxispad}}}
  \pgfpathlineto{\pgfpointadd{\pgfpointadd{\pgfplotspointrelaxisxy{0}{0}}{\pgfqpointxy{\xmaxcoord}{0}}}{\pgfqpoint{0.5\yafaxiswidth}{\yafaxispad}}}
  \pgfusepath{stroke}

}
\newcommand{\yafdrawyaxis}[2]{
  \pgfplotstransformcoordinatey{#1}\let\ymincoord=\pgfmathresult 
  \pgfplotstransformcoordinatey{#2}\let\ymaxcoord=\pgfmathresult 
  \pgfsetlinewidth{\yafaxiswidth} 
  \pgfsetcolor{yafaxiscolor}
  \pgfpathmoveto{\pgfpointadd{\pgfpointadd{\pgfplotspointrelaxisxy{0}{0}}{\pgfqpointxy{0}{\ymincoord}}}{\pgfqpoint{\yafaxispad}{-0.5\yafaxiswidth}}}
  \pgfpathlineto{\pgfpointadd{\pgfpointadd{\pgfplotspointrelaxisxy{0}{0}}{\pgfqpointxy{0}{\ymaxcoord}}}{\pgfqpoint{\yafaxispad}{0.5\yafaxiswidth}}}
  \pgfusepath{stroke}
}

\pgfplotscreateplotcyclelist{yaf}{%
{yafcolor1,mark options={scale=0.75},mark=o}, 
{yafcolor2,mark options={scale=0.75},mark=square},
{yafcolor3,mark options={scale=0.75},mark=triangle},
{yafcolor4,mark options={scale=0.75},mark=o},
{yafcolor5,mark options={scale=0.75},mark=o},
{yafcolor6,mark options={scale=0.75},mark=o},
{yafcolor7,mark options={scale=0.75},mark=o},
{yafcolor8,mark options={scale=0.75},mark=o}} 

\pgfplotsset{axis y line=left, axis x line=bottom,
  tick align=outside,
  compat = 1.3,
  tickwidth=\yafticklen,
  clip = false,
  every axis title shift = 0pt,
    x axis line style= {-, line width = 0pt, opacity = 0},
    y axis line style= {-, line width = 0pt, opacity = 0},
    x tick style= {line width = \yafaxiswidth, color=yafaxiscolor, yshift = \yafaxispad},
    y tick style= {line width = \yafaxiswidth, color=yafaxiscolor, xshift = \yafaxispad},
    x tick label style = {font=\scriptsize, yshift = \yaftlpad},
    y tick label style = {font=\scriptsize, xshift = \yaftlpad},
    every axis y label/.style = {at = {(ticklabel cs:0.5)}, rotate=90, anchor=center, font=\scriptsize, yshift = -\yaflabelpad},
    every axis x label/.style = {at = {(ticklabel cs:0.5)}, anchor=center, font=\scriptsize, yshift = \yaflabelpad},
    x tick label style = {font=\scriptsize, yshift = 1pt},
    grid = major,
    major grid style  = {dash pattern = on 1pt off 3 pt},
  every axis plot post/.append style= {line width=\yafaxiswidth} ,
  legend cell align = left,
  legend style = {inner sep = 1pt, cells = {font=\scriptsize}},
  legend image code/.code={%
    \draw[mark repeat=2,mark phase=2,#1] 
    plot coordinates { (0cm,0cm) (0.15cm,0cm) (0.3cm,0cm) };%
  } 
}

\tikzset{
  on each segment/.style={
    decorate,
    decoration={
      show path construction,
      moveto code={},
      lineto code={
        \path [#1]
        (\tikzinputsegmentfirst) -- (\tikzinputsegmentlast);
      },
      curveto code={
        \path [#1] (\tikzinputsegmentfirst)
        .. controls
        (\tikzinputsegmentsupporta) and (\tikzinputsegmentsupportb)
        ..
        (\tikzinputsegmentlast);
      },
      closepath code={
        \path [#1]
        (\tikzinputsegmentfirst) -- (\tikzinputsegmentlast);
      },
    },
  },
  mid arrow/.style={postaction={decorate,decoration={
        markings,
        mark=at position .6 with {\arrow[#1]{stealth}}
      }}},
}

\newcommand{\spara}[1]{\smallskip\noindent{\bf #1}}
\newcommand{\mpara}[1]{\medskip\noindent{\bf #1}}
\newcommand{\para}[1]{\noindent{\bf #1}}

\newtheorem{theorem}{Theorem}[section]
\newtheorem{observation}{Observation}[section]
\newtheorem{corollary}{Corollary}[section]
\newtheorem{proposition}{Proposition}[section]
\newtheorem{lemma}{Lemma}[section]
\newtheorem{remark}{Remark}[section]

\newtheorem{problem}{Problem}
\newtheorem{definition}{Definition}

\newcommand{\fpr}[1]{\mathopen{}\left(#1\right)}

\newcommand{\dispfunc}[2]{%
  \ensuremath{%
    \ifthenelse{\equal{\noexpand#2}{}}%
	     {#1}%
		      {{#1}\fpr{#2}}}}

\DeclareMathAlphabet{\pazocal}{OMS}{zplm}{m}{n}

\newcommand{\bigO}{\ensuremath{\mathcal{O}}\xspace}
\newcommand{\np}{\ensuremath{\mathsf{NP}}\xspace}
\newcommand{\p}{\ensuremath{\mathsf{P}}\xspace}
\newcommand{\wone}{\ensuremath{\mathsf{W[1]}}\xspace}

\newcommand{\fpt}{\ensuremath{\textsf{FPT}}\xspace}

\newcommand{\wtwo}{\ensuremath{\mathsf{W[2]}}\xspace}

\newcommand{\SETH}{\ensuremath{\textsf{SETH}}\xspace}

\newcommand{\GAPETH}{\ensuremath{\textsf{Gap\text{-}ETH}}\xspace}
\newcommand{\gapeth}{\ensuremath{\textsf{Gap\text{-}ETH}}\xspace}
\newcommand{\todo}[1]{{\textcolor{red}{TODO: #1}}}

\newcommand{\poly}{\textsf{poly}\xspace}
\newcommand{\wrt}{\text{with respect to}\xspace}
\newcommand{\divkmedian}{\ensuremath{\text{\sc Div-}k\text{\sc-Median}}\xspace}
\newcommand{\divkmeans}{\ensuremath{\text{\sc Div-}k\text{\sc-Means}}\xspace}
\newcommand{\divkcenter}{\ensuremath{\text{\sc Div-}k\text{\sc-Center}}\xspace}
\newcommand{\divksupplier}{\ensuremath{\text{\sc Div-}k\text{\sc-Supplier}}\xspace}

\newcommand{\fairkmedian}{\ensuremath{\text{\sc Fair-\-}k\text{\sc-Median}}\xspace}
\newcommand{\fairkmeans}{\ensuremath{\text{\sc Fair-\-}k\text{\sc-Means}}\xspace}
\newcommand{\fairkcenter}{\ensuremath{\text{\sc Fair-\-}k\text{\sc-Center}}\xspace}
\newcommand{\fairksupplier}{\ensuremath{\text{\sc Fair-\-}k\text{\sc-Supplier}}\xspace}

\newcommand{\fairrangekmedian}{\ensuremath{\text{\sc Fair-}k\text{\sc-Median}}\xspace}
\newcommand{\fairrangekmeans}{\ensuremath{\text{\sc Fair-}k\text{\sc-Means}}\xspace}
\newcommand{\fairrangekcenter}{\ensuremath{\text{\sc Fair-}k\text{\sc-Center}}\xspace}
\newcommand{\fairrangeksupplier}{\ensuremath{\text{\sc Fair-}k\text{\sc-Supplier}}\xspace}

\newcommand{\kmeans}{\ensuremath{k\text{\sc{-Means}}}\xspace}
\newcommand{\kcenter}{\ensuremath{k\text{\sc{-Center}}}\xspace}

\newcommand{\ksupplier}{\ensuremath{k\text{\sc{-Supplier}}}\xspace}

\newcommand{\reqsat}{\ensuremath{\text{\sc Div-(}{\vec{\alpha},\vec{\beta}}\text{\sc{)-Sat}}}\xspace}

\newcommand{\vertexcover}{\ensuremath{\text{\sc Ver\-tex\-Cov\-er}}\xspace}
\newcommand{\dominatingset}{\ensuremath{k\text{\sc -Dom\-inat\-ing\-Set}}\xspace}
\newcommand{\hittingset}{\ensuremath{\text{\sc Hitt\-ing\-Set}}\xspace}

\newcommand{\kmedianppm}{\ensuremath{k\text{\sc-Median-}p\text{\sc-\-PM}}\xspace}
\newcommand{\kmeansppm}{\ensuremath{k\text{\sc-Mea\-ns-}p\text{\sc-\-PM}}\xspace}

\newcommand{\ksupplierppm}{\ensuremath{k\text{\sc-Supp\-lier-}p\text{\sc-\-PM}}\xspace}

\newcommand{\kmediankpm}{\ensuremath{k\text{\sc-Med\-ian-\-}k\text{\sc-\-PM}}\xspace}
\newcommand{\kmeanskpm}{\ensuremath{k\text{\sc-Mea\-ns-\-}k\text{\sc-\-PM}}\xspace}

\newcommand{\ksupplierkpm}{\ensuremath{k\text{\sc-Supp\-lier-\-}k\text{\sc-\-PM}}\xspace}

\newcommand{\kmedianpm}{\ensuremath{k\text{\sc-Med-}k\text{\sc-PM}}\xspace}
\newcommand{\kmedian}{\ensuremath{k\text{\sc{-Med\-ian}}}\xspace}

\newcommand{\cost}{\ensuremath{\text{\sf cost}}\xspace}
\newcommand{\impr}{\ensuremath{\text{\sf improv}}\xspace}
\newcommand{\coreset}{\ensuremath{\text{\sf core-set}}\xspace}

\newcommand{\pattern}{constraint pattern\xspace}
\newcommand{\patternset}[1]{E(#1)\xspace}
\newcommand{\charpart}{\mathcal{P}\xspace}
\newcommand{\charvec}{\vec{\chi}\xspace}

\newcommand{\divkins}{\ensuremath{((U,d),F,C,\Gcal,\alphavec,\betavec,k)}\xspace}

\newcommand{\kt}{k+t}

\newcommand{\smallball}{\Pi\xspace}

\def\DEBUG{true}

\ifdefined\DEBUG
\def\attn#1{\textcolor{red}{#1}} 
\def\todo#1{\textcolor{red}{TODO: [#1]}} 
\newcommand{\attention}[1]{\textcolor{red}{** #1 **}}

\def\rem#1{\marginpar{\raggedright\scriptsize #1}}
\newcommand{\agir}[1]{\rem{\textcolor{RedViolet}{$\bullet$ Aris: #1}}}
\newcommand{\amtr}[1]{\rem{\textcolor{BurntOrange}{$\bullet$ A: #1}}}
\newcommand{\borr}[1]{\rem{\textcolor{WildStrawberry}{$\bullet$ B: #1}}}
\newcommand{\sthr}[1]{\rem{\textcolor{Green}{$\bullet$ S: #1}}}

\newcommand{\sandy}[1]{\rem{\textcolor{violet}{$\bullet$ Sa: #1}}}
\else
\def\attn#1{} 
\def\todo#1{} 
\def\ameet#1{} 
\newcommand{\attention}[1]{}

\newcommand{\ftmr}[1]{}
\newcommand{\faab}[1]{}
\newcommand{\sanr}[1]{}
\newcommand{\agir}[1]{}
\newcommand{\borr}[1]{}
\newcommand{\amtr}[1]{}
\newcommand{\kamr}[1]{}
\newcommand{\sthr}[1]{}
\newcommand{\roor}[1]{}
\newcommand{\danr}[1]{}
\newcommand{\sandy}[1]{}
\newcommand{\roohani}[1]{}
\fi
\usepackage[sort&compress, numbers]{natbib}

\title{Diversity-aware clustering: computational complexity and approximation algorithms
}

\author{
  Suhas Thejaswi\\
  Max Planck Institute for Software Systems \\
  Kaiserslautern, Germany\\
  \texttt{thejaswi@mpi-sws.org} \\
  \And
  Ameet Gadekar \\
  CISPA Helmholtz Center for Information Security\\
  Saarbr{\"u}cken, Germany\\
  \texttt{ameet.gadekar@cispa.de}
  \AND
  Bruno Ordozgoiti\\
  London, United Kingdom\\
  \texttt{bruno.ordozgoiti@gmail.com}\\
  \And
  Aristides Gionis \\
  KTH Royal Institute of Technology\\
  Stockholm, Sweden \\
  \texttt{argioni@kth.se} \\
}

\begin{document}
\maketitle

\begin{abstract}
In this work, we study diversity-aware clustering problems where the data points are associated with multiple attributes resulting in intersecting groups.  A clustering solution needs to ensure that the number of chosen cluster centers from each group should be within the range defined by a lower and upper bound threshold for each group, while simultaneously minimizing the clustering objective, which can be either $k$-median, $k$-means or $k$-supplier. 
We study the computational complexity of the proposed problems, offering insights into their \np-hardness, polynomial-time inapproximability, and fixed-parameter intractability. 
We present parameterized approximation algorithms with approximation ratios $1+ \frac{2}{e} + \epsilon \approx 1.736$, $1+\frac{8}{e} + \epsilon \approx 3.943$, and $5$ for diversity-aware $k$-median, diversity-aware $k$-means and diversity-aware $k$-supplier, respectively. 
Assuming \gapeth, the approximation ratios are tight for the diversity-aware $k$-median and diversity-aware $k$-means problems. 
Our results imply the same approximation factors for their respective fair variants with disjoint groups---fair $k$-median, fair $k$-means, and fair $k$-supplier---with lower bound requirements.
\end{abstract}


\section*{Statement on ethics and integrity.}

\para{Ethics approval.} Our work is theoretical in nature, and our work does not involve any human subject study and/or crowd-sourcing, as such approval from ethics committee and/or an institutional review board is not mandated for this research. 

\para{Funding.} Suhas Thejaswi is supported by the European Research Council (ERC) under the European Union'{}s Horizon $2020$ research and innovation program ($945719$) and the European Unions'{}s SoBigData++ Transnational Access Scholarship.
%
%
Aristides Gionis is supported by the ERC Advanced Grant REBOUND ($834862$), the European Union'{}s Horizon $2020$ research and innovation project SoBigData++ (871042), and the Wallenberg AI, Autonomous Systems and Software Program (WASP) funded by the Knut and Alice Wallenberg Foundation.

\para{Conflict of interest.} Authors declare no conflict of interest.


\newpage

\section{Introduction}
{\em Diversity} is an essential design choice across numerous real-world contexts, spanning social environments~\cite{healey2019diversity}, organizational structures~\cite{shore2018inclusive}, and demographic studies~\cite{zinn1996theorising}. Embracing diversity entails acknowledging and incorporating multifaceted characteristics within groups. Supporting diversity becomes important when addressing real-world challenges, particularly in scenarios where {\em intersectionality}---the interconnected nature of social categorizations such as gender, ethnicity, religion, socio-economic status and sexual orientation---plays a pivotal role~\cite{crenshaw2013demarginalizing,runyan2018intersectionality}.

Consider the task of constituting a representative committee that accurately mirrors the demography of a broader population. In the pursuit of diversity and fairness, it is imperative to ensure rep\-re\-sen\-ta\-tion from various groups based on their gender, ethnicity, and economic sta\-tus, among others~\cite{fish1993reverse}. In reality, individuals belong to multiple social categories. For example, a person could be of a specific gender, ethnic background and economic group. Focusing solely on gender may overlook the representation with respect to ethnicity and economic status. Furthermore, considering the groups independently could neglect the intersectionality of these identities.
For instance, Kearns et al.~\cite{kearns2019empirical} show that inter\-sectional subgroups may experience increased algorithmic harm.
Notably, the outcome of algorithmic classifiers have been observed to differ significantly between black and white~females~\cite{kasy2021fairness}.

The task of choosing a diverse committee can be formulated as a clustering problem:
the distance between individuals serves as a measure of (dis)similarity of their viewpoints, 
while the attributes such as gender, ethnicity, and economic status define group affiliations. 
The goal is to identify a subset of individuals that ensures sufficient representation of all groups while minimizing the distance between the chosen committee and the broader population.
In this process, each individual is associated with their closest representative in the committee, forming clusters where committee members act as cluster centers.
The optimization task aims to strike a balance between group representation and overall cohesion in the committee-formation process.

Currently, algorithms typically attempt to ensure fairness by addressing representation for each group independently, and do not consider the complexity introduced by intersecting attributes. Consequently, fairness measures applied separately to each group may not capture the nuanced ways in which biases manifest themselves when multiple attributes are considered together, as highlighted in earlier works~\cite{kasy2021fairness,kearns2019empirical,ghosh2021characterizing,hoffmann2019fairness,kong2022are}. 
This paper focuses on the exploration of diversity-aware clustering problems, where attributes associated with entities/individuals form intersecting groups. 
Our novel approach supports diversity in data-clustering tasks, where the cluster centers encapsulate intersectionality.

More formally, in diversity-aware clustering we are given a set of data points in a metric space, a collection of potentially intersecting groups of data points, a specified lower and upper bound indicating the minimum and the maximum number of points that can be selected as cluster centers from each group, and the desired number of cluster centers to be chosen.
The task is to find a subset of data points of the desired size, ensuring that the number of points chosen from each group falls within the specified lower and upper bounds, while minimizing the clustering objective.
We focus on three common clustering objectives: \emph{$k$-median}, \emph{$k$-means}, and \emph{$k$-supplier}. The $k$-median ($k$-means) objective seeks to minimize the sum of (squared) distances between data points and their closest cluster center. 
The $k$-supplier objective minimizes the maximum distance between the data points and their closest cluster center.\footnote{In relevant literature, the data points are categorized as clients or facilities, and the two sets may overlap. A distinction is made between the $k$-center and the $k$-supplier problem based on the selection of cluster centers. In the $k$-center problem, all data points are considered potential facilities, and consequently, are eligible to be chosen as cluster centers, whereas in the $k$-supplier problem, only a subset of data points are considered as facilities~\cite{hochbaum1986unified}.} Building upon these clustering objectives, we define diversity-aware $k$-median, diversity-aware $k$-means, and diversity-aware $k$-supplier problems, respectively.
For a precise formulation of the problems, see Section~\ref{sec:problem-definition}.

Diversity-aware clustering optimizing the $k$-median objective with intersecting (facility) groups was introduced by Thejaswi et al.~\cite{thejaswi2021diversity}. In this setting, constraints are imposed on the minimum number of cluster centers (lower bound) to be chosen from each group. They established the polynomial-time inapproximability for the general variant with intersecting groups and presented polynomial-time approximation algorithms for the case with disjoint groups. In a subsequent work, Thejaswi et al.~\cite{thejaswi2022clustering} studied the variant with intersecting groups, offering complexity results and tight parameterized approximation algorithms for diversity-aware $k$-median and diversity-aware $k$-means (with lower bound constraints).
This paper extends the work of Thejaswi et al.~\cite{thejaswi2022clustering}, expanding the scope to include \emph{both lower and upper bound requirements} as well as extending the approach to accommodate $k$-median, $k$-means, and $k$-supplier clustering objectives. 
Additionally, we establish the optimality of the presented algorithms based on standard complexity-theoretic assumptions. Table~\ref{table:summary} summarizes our computational complexity and algorithmic results.
Specifically, our contributions are as follows:
\squishlist
\item We study the computational complexity of the proposed problems, offering insights into \np-hardness, polynomial-time inapproximability, and fixed parameter intractability (Section~\ref{sec:complexity}).\footnote{Fixed parameter intractability rules out algorithms whose exponential running time can be bounded by certain (natural) parameters of the problem. See Appendix~\ref{appendix:sec:parameterized} for a precise definition.} 

\item For diversity-aware $k$-median, diversity-aware $k$-means, and diversity-aware $k$-supplier problems, we present parameterized approximation algorithms with respect to the problem parameters---the number of cluster centers $k$ and the number of groups $t$---achieving approximation ratios $1 + \frac{2}{e}+\epsilon$, $1 + \frac{8}{e}+\epsilon$, and $5$, respectively (Theorems~\ref{theorem:mainfptapx} and~\ref{thm:maindivksupplier}).

\item We establish that the approximation ratios presented for diversity-aware $k$-median and diversity-aware $k$-means are optimal for any parameterized approximation algorithm \wrt parameters $k$ and $t$. This assertion is based on the Gap Exponential Time Hypothesis (\GAPETH) (Theorems~\ref{theorem:mainfptapx} and~\ref{thm:maindivksupplier}).
\squishend

The subsequent sections of this paper are structured as follows: In Section~\ref{sec:related} we discuss related works. The problem formulations are detailed in Section~\ref{sec:problem-definition}, and Section~\ref{sec:complexity} is dedicated to the discussion of computational complexity results. In Section~\ref{sec:parameterized}, we present approximation algorithms for diversity-aware clustering. Finally, in Section~\ref{sec:conclustion} we present concluding remarks and directions for future work.
\section{Related work}
\label{sec:related}

Our work builds on existing literature on clustering as well as on recent work on algorithmic fairness. Data clustering is a fundamental problem in computer science and has been studied extensively across several algorithmic regimes,  including computational complexity~\cite{guha1998greedy,cohen2019tight},  exact exponential algorithms~\cite{fomin2022exact},  approximation algorithms~\cite{charikar1999constant,arya2001local,cohen2022improved},  pseudo-approximation algorithms~\cite{li2016approximating}, and  parameterized algorithms~\cite{cohen2019tight,goyal2023tight}. The study of clustering with fairness aspects has also garnered considerable attention~%
%
\cite{hajiaghayi2012local,kleindessner2019fair,ghadiri2021socially,goyal2023tight,chen2016matroid,krishnaswamy2011matroid,hotegni2023approximation,abbasi2023parameterized,abbasi2024parameterized}. 
While we refer to the problem we study as \emph{diversity-aware clustering}, 
our formulation is also related to \emph{fair clustering}, 
which has received considerable attention in the literature recently.
Since the related work on clustering and fair clustering is extensive, 
here we only discuss the results most relevant to our work.\footnote{For a survey on fairness in clustering see Chhabra et al.~\cite{chhabra2021overview} and tutorial resources by Brubach et al.~\cite{brubach2022fairness}.}

Charikar et al.~\cite{charikar2002constant} presented the first constant-factor approximation for $k$-median in metric spaces, which was improved to $(3+\epsilon)$-approximation by Arya et al.~\cite{arya2004local} using a local-search heuristic. 
More recently, Cohen-Addad et al.~\cite{cohen2022improved} refined the local-search heuristic
to obtain a $(2.836+ \epsilon)$ approximation. 
The best-known approximation ratio for metric instances stands at $2.675$ due to Byrka et al.~\cite{byrka2014improved}. 
For $k$-means, Kanungo et al.~\cite{kanungo2004local} devised a $(9+ \epsilon)$ approximation algorithm, 
later improved to $6.357$ by Ahmadian et al.~\cite{ahmadian2019better}.
On the other hand, $k$-median and $k$-means problems are \np-hard to approximate to a factor less than $1+\frac{2}{e}$ and $1 + \frac{8}{e}$, respectively~\cite{guha1998greedy}. Bridging the gap between the lower-bound of approximation and the achievable approximation ratio remains a well-known open problem.

In the realm of fixed-parameter tractability~(\fpt), finding an optimal solution to $k$-median, $k$-means, $k$-center, and $k$-supplier problems is known to be \wtwo-hard \wrt parameter~$k$~\cite{guha1998greedy,hochbaum1986unified}. Unless explicitly stated otherwise, all \fpt algorithms discussed in this section are parameterized by $k$.
Recently, Cohen-Addad et al.~\cite{cohen2019tight} presented \fpt approximation algorithms with approximation ratios $1+\frac{2}{e}+\epsilon$ and $1+\frac{8}{e}+ \epsilon$ for $k$-median and $k$-means, respectively, which are tight assuming the \gapeth.
Assuming $\fpt \neq \wtwo$, for any $\epsilon>0$, Goyal and Jaiswal~\cite[Theorem~1, Theorem~2]{goyal2023tight} 
showed that the $k$-center and the $k$-supplier problems cannot be approximated to $(2-\epsilon)$ and $(3-\epsilon)$ factors, respectively, in \fpt time.
However, when the data points belong to Euclidean space---prevalent in many practical applications---a $(3-\epsilon)$-approximation algorithm is known for $k$-supplier in \fpt time~\cite{abbasi2024parameterized}.

The notion of fairness in algorithm design has gained significant traction in recent years~\cite{matakos2024fair,chierichetti2017fair,samadi2018fairpca}, and various fairness-aware variants of clustering formulations have been introduced, each catering to different fairness notions~\cite{hajiaghayi2012local,kleindessner2019fair,ghadiri2021socially,goyal2023tight,chen2016matroid,krishnaswamy2011matroid,hotegni2023approximation,abbasi2023parameterized,abbasi2024parameterized,chen2024approximation}. Our emphasis in this work is on cluster-center fairness, ensuring that the number of cluster centers chosen from each group in within the lower and upper bound thresholds.
A notable problem in fair clustering is the red-blue median problem~\cite{hajiaghayi2012local}, in which the facilities are colored red or blue, and a solution may contain only up to a specified number of facilities (upper-bound) of each color. This formulation was generalized to the matroid-median problem~\cite{krishnaswamy2011matroid}, where solutions must be independent sets of a matroid. Constant-factor approximation algorithms for both of these problems are known~\cite{hajiaghayi2012local,krishnaswamy2011matroid,krishnaswamy2018constant}.
Thejaswi et al.~\cite{thejaswi2021diversity} studied the fair-$k$-median problem with disjoint facility groups, where a minimum number of cluster centers  (lower-bound constraint) must be selected from each group. They established hardness results and presented approximation algorithms. 
Recently, Zhang et al.~\cite{zhang2024towards} gave a multi-swap local-search algorithm that achieves $(4k+5)$-factor approximation for fair-$k$-median with lower bound requirements, although the running time of their method is not fully detailed. 

Kleindessner et al.~\cite{kleindessner2019fair} investigated fair-$k$-center, where an exact number of cluster centers must be chosen from each group. They propose a constant-factor pseudo-polynomial time algorithm, which was improved to a polynomial-time $3$-approximation using maximal matching by Jones et al.~\cite{jones2020fair}. In both studies, they focus on scenarios where the set of clients and facilities are identical, and the constraints on the number of cluster centers from each group are exact requirements (not lower-bound constraints), making the problem relatively easier. Chen et al.~\cite{chen2024approximation}, study the more general problem of fair-$k$-supplier and present a $5$-approximation algorithm, based on maximal matching. More recently, building upon our work, Gadekar et al.~\cite{gadekar2025fair} gave a $3$ approximation algorithm for the fair $k$-supplier problem. Their algorithm runs in polynomial time when facility groups are disjoint and $\fpt(k+t)$ time when the groups intersect, considering both lower and upper bound requirements.
Recently, Hotegni et al.~\cite{hotegni2023approximation} presented a polynomial-time constant factor approximation algorithm for the fair range clustering problem with an $\ell_p$-norm objective, including both lower and upper bound requirements. 
For \fpt algorithms, as mentioned before, these problems remain \wtwo-hard. However, Abbasi et al.~\cite{abbasi2023parameterized} design an algorithmic framework that yields a $(1+\epsilon)$-\fpt approximation algorithm for several (fair) clustering variants across several metric spaces. This framework unifies previous results and extends the results to more complex clustering objectives, such as socially-fair clustering, under commonly studied metric spaces such as Euclidean and planar graphs.

\smallskip
Next, we formally introduce the diversity-aware clustering problems before continuing to study their computational complexity.
\section{Diversity-aware clustering problems} 
\label{sec:problem-definition}

In a $k$-clustering problem, we are given a metric space $(U,d)$, a set of clients~$C \subseteq U$, a set of facilities $F \subseteq U$, and an integer $k \in \ZZ_+$. The task is to find a subset of facilities $S \subseteq F$ of size $k$, which minimizes a certain clustering objective $\phi(C,S)$. The clustering problems $k$-median, $k$-means, and $k$-supplier are defined when the clustering objectives are
$\sum_{c \in C} d(c,S)$, $\sum_{c \in C} d(c,S)^2$, and $\max_{c \in C} d(c,S)$, respectively. Here, $d(c,S)$ denotes the minimum distance between~$c$ and~$S$, \ie, $d(c,S) = \min_{s\in S} d(c,s)$.

For the diversity-aware clustering problem, we consider a collection of possibly intersecting facility groups $\Gcal=\{G_i\}_{i \in [t]}$ for some $t \in \ZZ_+$, where each subset $G_i \subseteq F$ corresponds to a facility group. We are given $\alphavec=(\alpha_i)_{i\in [t]}$ and $\betavec=(\beta_i)_{i \in [t]}$) where each element $\alpha_i$ and $\beta_i$ signifies lower and upper bound on the number of facilities of $G_i$ that must be included in a clustering solution~$S$.
%
We formally introduce the diversity-aware clustering problems below.

\smallskip
\begin{problem}[Diversity-aware clustering]
\label{problem:divclust}
We are given a metric space $(U,d)$, with clients $C \subseteq U$, facilities $F \subseteq U$, 
a collection $\Gcal=\{G_i\}_{i \in [t]}$ of $t$ subsets $G_i \subseteq F$, 
vectors $\alphavec=(\alpha_i)_{i \in [t]}$ and $\betavec=(\beta_i)_{i \in [t]}$, 
and a non-negative integer $k$. 
The goal is to find a subset $S \subseteq F$ of facilities such that:
\emph{($i$)}~$|S|=k$;
\emph{($ii$)}~$S$ satisfies the constraints $\alpha_i \leq |S\cap G_i|\leq \beta_i$ for all $i \in [t]$; and
\emph{($iii$)}~$S$ minimizes the objective $\phi(C,S)$.
The \emph{diversity-aware $k$-median} problem, the \emph{diversity-aware $k$-means} problem, and 
the \emph{diversity-aware $k$-supplier} problem are defined by setting the objective function $\phi(C,S)$ to be 
$\sum_{c \in C} d(c,S)$, $\sum_{c \in C} d(c,S)^2$, and $\max_{c \in C} d(c,S)$, respectively. 
\end{problem}

\smallskip
Throughout this paper, a distinction is made based on whether the facility groups are mutually disjoint or possibly intersecting. Importantly, when the groups $\Gcal=\{G_i\}_{i \in [t]}$ are disjoint, 
we refer to the corresponding problems as {\em fair $k$-median}, {\em fair $k$-means}, and {\em fair $k$-supplier}, 
when the objective function $\phi(C,S)$ is 
$\sum_{c \in C} d(c,S)$, $\sum_{c \in C} d(c,S)^2$, and $\max_{c \in C} d(c,S)$, respectively.

To understand better the effect of the constraints of the problem, 
decoupled from the clustering objectives, we introduce the {\em diversity-requirements satisfiability problem}, 
which specifically aims to find a subset of facilities $S \subseteq F$ of size $|S|=k$ that contain at least $\alpha_i$ and at most $\beta_i$ facilities from each group $G_i$, for every $i \in [t]$. This consideration is made disregarding the clustering objective $\phi(C,S)$. The formal definition of this problem is as follows.

\smallskip
\begin{problem}[Diversity-requirements satisfiability problem]
\label{problem:reqsat}
Given a set of facilities~$F$, a collection $\Gcal=\{G_i\}_{i \in [t]}$ of facility subsets $G_i \subseteq F$, with~$i \in [t]$, vectors $\alphavec=(\alpha_i)_{i \in [t]}$ and $\betavec = (\beta_i)_{i\in [t]}$ of lower and upper bound requirements, respectively, and a non-negative integer~$k \in \ZZ_+$. Is there a subset of facilities $S \subseteq F$ of size $|S|=k$ that satisfies the constraints $\alpha_i \leq |S\cap G_i| \leq \beta_i$, for all $i \in [t]$? 
\end{problem}

\smallskip
While devising algorithmic solutions, we transform diversity-aware clustering and fair-clustering problems into instances of the  standard $k$-median, $k$-means, and $k$-supplier problems  
with $p$-partition matroid constraints. 
The formal problem definition is as follows:

\smallskip
\begin{problem}[$k$-clustering with $p$-partition matroid constraints]
\label{problem:kclusteringpm}
Given a metric space $(U,d)$, with a set $C \subseteq U$ of clients, 
a collection $\Ecal=\{E_i\}_{i \in [p]}$ of disjoint facility groups $E_i \subseteq U$, with $i \in [p]$,
called a {\em $p$-partition matroid}, and a non-negative integer $k \in \ZZ_+$,
find a subset of facilities $S \subseteq F$ of size $|S|=k$, 
containing at most one facility from each group $E_i$, {such that} the clustering objective $\phi(C,S)$ is minimized.
Based on the distinct clustering objectives, we introduce the variants of 
$k$-median, $k$-means, and $k$-supplier problems with $p$-partition matroid constraints.
\end{problem}

\smallskip
To summarize the problems that we have defined in this section,  we refer to the diversity-aware clustering problems as \divkmedian, \divkmeans, and \divksupplier;  to the fair clustering problems as \fairrangekmedian, \fairrangekmeans, and \fairrangeksupplier; and  to clustering with $p$-partition matroid constraints as \kmedianppm, \kmeansppm, and \ksupplierppm. Furthermore, we refer to the diversity-requirements satisfiability problem as \reqsat. 
\section{Computational complexity} 
\label{sec:complexity}

In this section, we first explore the polynomial-time approximability of fair $k$-clustering, 
i.e., when the facility groups are disjoint, 
and then delve into the computational complexity of diversity-aware $k$-clustering, 
i.e., focusing on intersecting facility groups. 
Specifically, we establish the \np-hardness and \wone-hardness with respect to various parameters. 
Further, we show inapproximability to any multiplicative factor even when exponential time 
with respect to various parameters is allowed.
Our computational-complexity results are crucial 
for appreciating the choice of the proposed algorithmic solutions.

\subsection{Polynomial-time approximability}

We discuss the fair-clustering problem, as defined in the previous section,  
namely, it is assumed that the facility groups are disjoint. 
This problem was defined by \citet{thejaswi2021diversity}, 
who further considered only lower-bound requirements, 
i.e., at least $\alpha_i$ facilities should be selected from each group~$G_i$. 
Thus, when referring to the fair-clustering problem in this section
we assume only lower-bound requirements.
%

\citet[Section~4.1]{thejaswi2021diversity} noted that the \fairkmedian  problem
can be approximated to a constant factor, 
as it reduces to the matroid-median problem, 
and thus, using existing methods for the latter problem. 
Such a reduction can be extended to \fairkmeans 
via the use of existing methods for matroid-means problem.
%
The best-known approximation for matroid median 
is $(7.081 + \epsilon)$~\citep[Section~6]{krishnaswamy2018constant},  
which is applicable to \divkmedian, as well. 
For two groups only, a $(5 + \epsilon)$-approximation is possible using multi-swap local search 
\citep[Section~4.2]{thejaswi2021diversity}.
More recently, \citet{zhang2024towards} obtained a $(4t + 5)$-approximation 
for \fairkmedian with $t$ groups, using $(t+1)^2$-swap local search.

The best-known approximation for matroid means 
is $64$~\citep[Thorem~7]{zhao2023improved}, 
yielding a $64$-approximation algorithm for \fairkmeans~\citep[Section~4.1]{thejaswi2021diversity}.
\citet{jones2020fair} gave a $3$-approximation algorithm for \fairkcenter, applicable when facilities and clients are identical, and exactly $\alpha_i$ centers are chosen from each group $G_i$.
%
Recently, \citet{chen2024approximation} presented a polynomial-time $5$-approximation algorithm for 
\fairksupplier, 
considering non-identical sets of clients and facilities, and selecting exactly $\alpha_i$ centers from each group $G_i$.

It is unclear whether algorithms designed for handling lower-bound or exact requirements can be adapted to solving 
fair  clustering with both lower- and upper-bound requirements.\footnote{The reduction by \citet[Section~4.1]{thejaswi2021diversity} does not hold when both upper- and lower-bound requirements are present.} 
Recently, \citet{hotegni2023approximation} presented a polynomial-time constant-factor approximation algorithm 
for fair clustering with \(\ell_p\)-norm objectives, addressing 
\fairrangekmedian, \fairrangekmeans, and \fairrangekcenter with identical sets of clients and facilities.

\subsection{Hardness and polynomial-time inapproximability}
Diversity-aware  clustering is an amalgamation of two independent problems. The first problem, referred as the diversity-requirements satisfiability problem (\reqsat), seeks to find a subset of facilities $S \subseteq F$ with cardinality $|S|=k$ that satisfies constraints $\alpha_i \leq |S \cap G_i| \leq \beta_i$, for all $i \in [t]$ (for a precise definition see Problem~\ref{problem:reqsat}). The second problem pertains to minimizing the clustering objective $\phi(C,S)$.
When the number of groups $t = 1$ with lower bound $\alpha_1 = 1$ and upper bound $\beta_1 = k$, \divkmedian, \divkmeans, and \divksupplier are equivalent to \kmedian, \kmeans, and \ksupplier, respectively. Therefore, the \np-hardness of the latter problems implies the \np-hardness of their diversity-aware counterparts. Note that the \np-hardness refers to their respective decision variants. 

\begin{observation} \label{proposition:hardness:1}
Problems \divkmedian, \divkmeans, and \divksupplier are \np-hard.
\end{observation}
\smallskip

A stronger result can be established, namely, \divkmedian, \divkmeans and \divksupplier are \np-hard to approximate to any multiplicative factor in poly\-nomial time. This inapproximability result can be derived from the fact that finding a feasible solution for \reqsat is \np-hard. Thejaswi et al.~\cite{thejaswi2021diversity} originally proved Lemma~\ref{lemma:nphardness} and Theorem~\ref{theorem:inapproximability} focusing only on lower bound requirements. For com\-pleteness, we present these proofs considering both lower and upper bound requirements.

\begin{lemma} \label{lemma:nphardness}
Problem \reqsat is \np-hard.
\end{lemma}
\begin{proof}
Given an instance $(G=(V,E),k)$ of the dominating set problem (\dominatingset),\footnote{In the dominating set problem, a graph $G=(V,E)$ and an integer $k$ are given, the objective is to decide if there exists a subset $S \subseteq V$ of vertices of cardinality $|S|=k$ such that, for every vertex $u \in V$, either $u$ is in $S$ or at least one neighbor of $u$ in $S$, that is, $u \in S$ or $N(u)\cap S \neq \emptyset$.The neighbors of vertex $u \in V$ is denoted as $N(u)$, that is, $N(u)=\{v: (u,v) \in E\}$.} 
the construction of \reqsat instance is as follows: set $U=F=V$, define the distance function $d(u,v)=1$ for all $(u,v) \in U \times U$ and $u \neq v$. The facility groups $\Gcal=\{G_u\}_{u \in V}$ are defined as $G_u = \{u\} \cup N(u)$ with lower and upper bound requirements $\alphavec=\mathbf{1}_{|V|}$ and $\betavec=\mathbf{1}_{|V|}$, where $\alpha_u, \beta_u=1$ for all $u \in V$.\footnote{A vector of length $\ell$ with all entries $1$ is denoted as $\mathbf{1}_\ell$.} The construction is polynomial in the size of the input.
Suppose $S \subseteq F$ is a feasible solution for \reqsat. It's evident from the construction that $S$ is a dominating set because $|G_u \cap S| \geq 1$, indicating that $S$ intersects $\{u\} \cup N(u)$ for all $u\in V$. If $S \subseteq F$ is a feasible solution for \dominatingset, then $S \cap G_u \neq \emptyset$ due to the fact that $S$ is a dominating set and $G_u = u \cup N(u)$. This completes our proof.
\end{proof}

The \np-hardness of \reqsat implies the polynomial-time inapproximability of \divkmedian, \divkmeans, and \divksupplier to any multiplicative factor, as stated in the following theorem.

\smallskip
\begin{theorem}
\label{theorem:inapproximability}
Assuming $\p \neq \np$, there exists no polynomial-time algorithm to approximate \divkmedian, \divkmeans, and \divksupplier to any multiplicative factor.
\end{theorem}
\begin{proof}
Assuming the existence of a polynomial-time algorithm that can approximate any feasible instance of \divkmedian (\divkmeans, \divksupplier, resp.) to a constant factor, we can use it to solve \dominatingset in polynomial time. This contradiction is obtained through the reduction described in Lemma~\ref{lemma:nphardness}. The approximate solution for \divkmedian (\divkmeans, \divksupplier, resp.) is also a valid solution for \reqsat, and consequently for \dominatingset, proving the polynomial-time inapproximability of \divkmedian (\divkmeans, \divksupplier resp.).
\end{proof}

We continue to show that inapproximability results extend to specific cases with restricted input structures. The reduction in Theorem~\ref{theorem:inapproximability} confirms that \divkmedian (\divkmeans, \divksupplier resp.) remain inapproximable in polynomial-time to any multiplicative factor, even when every requirement is $1$ \ie, $\alphavec = \betavec = \mathbf{1}_t$. Furthermore, Theorem~\ref{theorem:inapproximability:2} establishes polynomial-time inapproximability even when all subsets in collection $\Gcal = \{G_i\}_{i \in [t]}$ have size precisely $2$, \ie, $|G_i|=2$ for all $i \in [t]$. Meanwhile, Theorem~\ref{theorem:inapproximability:3} indicates inapproximability when the underlying metric space is a tree metric.

\smallskip
\begin{theorem}
\label{theorem:inapproximability:2}
Assuming $\p \neq \np$, there exists no polynomial-time algorithm to approximate \divkmedian, \divkmeans, and \divksupplier to any multiplicative factor, even if all the subsets in $\mathcal{G}=\{G_i\}_{i \in [t]}$ have size $|G_i|=2$.
\end{theorem}
\begin{proof}
To establish this result, we reduce the vertex cover problem (\vertexcover) to \divkmedian.\footnote{In the vertex cover problem, a graph $G=(V,E)$ and an integer $k \in \ZZ_+$ are given, and the objective is to decide if there exists a vertex subset $S \subseteq V$ of size $|S| = k$ such that for every edge $(u,v) \in E$ at least one of the vertices $u$ or $v$ is in $S$. \vertexcover is known to be \np-hard~\cite{garey2002computers}.} 

Consider an instance $(G=(V,E), k)$ of \vertexcover with $n$ vertices, $m$ edges, and $k \leq |V|$. We construct an instance $\divkins$ of \divkmedian as follows: First, we set $U=C=F=V$. Then, for each edge $\{u,v\} \in E$ we construct a group $G_i = \{u, v\}$, $\Gcal = \{G_i\}_{i \in [m]}$. We set the lower and upper bound thresholds $\alphavec = \betavec = \mathbf{1}_m$, and finally, the distance $d(u,v)=1$, for all $(u,v) \in E$, and $d(u,v)=n+1$, for all $(u,v) \notin E$. The construction is polynomial in the size of the \vertexcover instance.

Let $S \subseteq V$ be a solution to \vertexcover. It is evident from the construction that $|S \cap G_i| \geq 1$ for each $G_i \in \mathcal{G}$ since each $G_i$ represents a set of vertices in an edge. Additionally, as $|S| \leq k$, the set $S$ also serves as a solution to \divkmedian. The argument in other direction follows a similar line of reasoning. 

The construction demonstrates that if there exists a polynomial-time approximation algorithm for \divkmedian, we can solve \vertexcover in polynomial time. However, this is unlikely under the premise that $\p \neq \np$. The same argument applies for \divkmeans and \divksupplier.
\end{proof}


A curious reader might ask whether the metric property of the problem instance, could help us design efficient solutions. After all, clustering problems on specific metrics, such as tree or Euclidean metrics, often admit efficient solutions and different hardness of approximation bounds~\cite{cohenaddad2022improved}. However, we caution against pursuing this direction, as Lemma~\ref{lemma:nphardness} holds regardless of whether the underlying distance function satisfies metric properties. To support this claim, we present the following Theorem.

\begin{theorem} \label{theorem:inapproximability:3}
Assuming $\p \neq \np$, there exists no polynomial-time algorithm to approximate \divkmedian, \divkmeans, and \divksupplier to any multiplicative factor, even if the underlying metric space $(U,d)$ is a tree metric. 
\end{theorem}
\begin{proof}
We present a proof by contradiction. By the seminal work of Bartal~\cite{bartal1996probabilistic}, a metric space $(U,d)$ can be embedded into a tree metric with at most $\log |U|$-factor distortion in distances. Consequently, \divkmedian (\divkmeans,\divksupplier, resp.) can be embedded into a tree metric with at most $\log |U|$-factor distortion in distances.
Assuming the existence of a polynomial-time $\gamma$-approximation algorithm for \divkmedian on a tree metric, there exists a $\gamma \cdot \log |U|$-approximation algorithm for any metric instance of \divkmedian. A similar argument holds for \divkmeans and \divksupplier.
However, the existence of such an algorithm contradicts our inapproximability result in Theorem~\ref{theorem:inapproximability}. Therefore, \divkmedian (\divkmeans, \divksupplier, resp.) are \np-hard to approximate to any multiplicative factor, even if the underlying metric space is a tree.
\end{proof}

The \np-hardness and polynomial-time inapproximability of diversity-aware  clustering, even on simpler problem instances such as tree metrics with group size two, challenges our hope of designing polynomial-time algorithms even for structured inputs. It prompts us to explore super-polynomial time algorithms that remain efficient when certain parameters of the problem are small.
Unfortunately, we show that  this is not the case for \divkmedian (\divkmeans, \divksupplier resp.) for several natural parameters of the problem.

\subsection{Parameterized complexity}
In essence, fixed parameter tractable (\fpt) algorithms restrict the exponential dependence of running time to specific parameters of the problem. When these parameters are small, the runtime of the algorithm is also small.\footnote{A problem specified by input $x$ and a parameter $k$ is {\em fixed-parameter tractable} (\fpt) with respect to $k$, if there exists an algorithm to solve every problem instance $(x,k)$ with running time of the form $f(k) |x|^{\bigO(1)}$, where $f(k)$ is function depending solely on the parameter $k$ and $|x|^{\bigO(1)} = \mathrm{poly}(|x|)$ is a polynomial function.}
There exists a class of problems that are believed to be not fixed parameter tractable with respect to certain parameters of the problem, which collectively form the $\mathbf{W}$-hierarchy.
A brief introduction to parameterized complexity and related terminology is presented in  Appendix~\ref{appendix:sec:parameterized}.\footnote{To maintain conciseness, an \fpt algorithm with respect to parameter $k$ is denoted as $\fpt(k)$. When the algorithm pertains to multiple parameters $k_1, \dots, k_\ell$, it is denoted as $\fpt(k_1+ \dots + k_\ell)$.}

The reduction presented by Guha and Khuller~\cite[Theorem~3.1]{guha1998greedy} implicitly establishes the \wtwo-hardness of \kmedian and \kmeans with respect to parameter $k$. Similarly, \ksupplier is known to be \wtwo-hard \wrt $k$ as it is a generalization of \dominatingset, which also exhibits \wtwo-hardness~\cite[Lemma~23.2.1]{downey2013fundamentals}. When the number of groups $t=1$, \divkmedian, \divkmeans, and \divksupplier are equivalent to \kmedian, \kmeans, and \ksupplier, respectively. Consequently, the former problems are also \wtwo-hard \wrt $k$. These findings are summarized in the following proposition.

\smallskip
\begin{proposition}
\label{proposition:hardness:wtwo}
Problems \divkmedian, \divkmeans, and \divksupplier are \wtwo-hard with respect to parameter $k$.
\end{proposition}

To reinforce, by combining the reduction presented in Lemma~\ref{lemma:nphardness} with the strong exponential time hypothesis (\SETH)~\cite{impagliazzo2001on}, we draw a compelling conclusion in the below corollary: focusing solely on the parameter $k$, the most viable approach for finding an optimal, or even an approximate solution for \divkmedian, \divkmeans, and \divksupplier, is a straightforward exhaustive-search algorithm.

\begin{corollary}
\label{corollary:inapproximability:4}
Assuming \SETH, for any $k \geq 3$ and $\epsilon >0$, there exists no $\bigO(|{F}|^{k - \epsilon})$ algorithm to solve \divkmedian, \divkmeans, and \divksupplier optimally. Furthermore, there exists no $\bigO(|{F}|^{k - \epsilon})$ algorithm to approximate \divkmedian, \divkmeans, and \divksupplier to any multiplicative factor.
\end{corollary}
\begin{proof}
\SETH implies that there is no $\bigO(|V|^{k-\epsilon})$ algorithm for \dominatingset, for any $\epsilon>0$. The reduction in Lemma~\ref{lemma:nphardness} creates an instance of \divkmedian where the set of facilities $F=V$ and groups $\Gcal=\{G_u = u \cup N(u)\}_{u \in V}$. Hence, any \fpt exact or approximation algorithm running in time $\bigO(|F|^{k-\epsilon})$ for \divkmedian can solve \dominatingset in $\bigO(|V|^{k-\epsilon})$ time, contradicting \SETH. The same argument holds for \divkmeans and \divksupplier.
\end{proof}

Given the \wtwo-hardness of diversity-aware  clustering problems with respect to $k$, it is natural to consider further relaxations on the running time of algorithms for the problem. An obvious question is whether we can approximate \divkmedian, \divkmeans, and \divksupplier in \fpt time with respect to $k$, if we are allowed to open, say, $f(k)$ facilities instead of $k$, for some function $f$. Unfortunately, this is also unlikely as \divkmedian, \divkmeans, and \divksupplier captures \dominatingset from the reduction in Lemma~\ref{lemma:nphardness}, and finding a dominating set even of size $f(k)$ is \wone-hard~\cite{karthik2019on}.

\begin{proposition}
\label{proposition:inapproximability:5}
For parameter $k$, problems \divkmedian, \divkmeans, and \divksupplier are $\wone$-hard to approximate to any multiplicative factor, even when permitted to open $f(k)$ facilities, for any polynomial function $f$.
\end{proposition}
\begin{proof}
Any $\fpt(k)$ algorithm achieving a multiplicative factor approximation for \divkmedian (\divkmeans, \divksupplier, resp.) needs to solve \reqsat to satisfy requirements. Since these requirements capture \dominatingset Lemma~\ref{lemma:nphardness}, it means we can solve \dominatingset in $\fpt(k)$ time, which is a contradiction. Finally, noting the fact that finding dominating set of size $f(k)$, for any function $f$, is also \wone-hard due to~\cite[Theorem~1.3]{karthik2019on} concludes the proof.
\end{proof}
A possible way forward is to identify other parameters of the problem to design \fpt algorithms to solve the problem optimally. As established earlier, \kmedian, \kmeans, and \kcenter are special cases of \divkmedian, \divkmeans, and \divksupplier, when $t=1$. This immediately rules out an exact \fpt algorithm for \divkmedian with respect to parameters $(\kt)$. Furthermore, we caution the reader against entertaining the prospects of other, arguably natural, parameters, such as the maximum lower bound $r = \max_{i\in[t]} r_i$ (ruled out by the relation $r_i \leq k$) and the maximum number of groups a facility can belong to $\mu =\max_{f \in F}(|{G_i}: f \in G_i, i \in [t]|)$, ruled out by the relaxation $\mu \leq t$.

\smallskip
\begin{proposition}
\label{proposition:introduction:3}
Problems \divkmedian, \divkmeans, and \divksupplier are $\wtwo$-hard with respect to parameters $(\kt)$, $(r+t)$ and $(\mu+t)$.
\end{proposition}

\medskip
The above intractability results thwart our hopes of solving diversity-aware  clustering optimally in \fpt time. However, in the next section, we show that \fpt approximation algorithms can be designed for these problems.
\section{Parameterized approximation algorithms} \label{sec:parameterized}
Throughout this section, our \fpt algorithm pertains to parameters $k$, the size of solution sought and $t$, the number of groups, and it is denoted as $\fpt(\kt)$.

First, we present an $\fpt(\kt)$ algorithm for \divkmedian, and further use these techniques to derive an algorithm for \divkmeans. In addressing both \divkmedian and \divkmeans, our algorithm relies on constructing coresets of a size $\bigO(k \log |U| \nu^{-2})$, for some $\nu > 0$. However, techniques for constructing a coreset of similar size for \divksupplier are presently not known. Hence, we present an alternative algorithmic approach to solve \divksupplier in Section~\ref{sec:algorithm:fpt-divksupplier}.

\subsection{Coresets for \divkmedian and \divkmeans} \label{sec:algorithm:coresets}

The high-level idea of coresets is to reduce the number of clients such that the distortion in the sum of distances is bounded by a factor. Precisely, for every $\nu > 0$, we reduce the client set $C$ to a weighted set $C'$ with size $|C'|= \bigO(\Gamma k \log|U|)$, where the distortion in distances is bounded by a multiplicative factor $(1 \pm \nu)$ and $\Gamma > 0$.
Formally, a coreset is defined as follows:
\begin{definition}[Coreset]
Given an instance $((U,d),C,F,k)$ of \kmedian (\kmeans resp.) and a constant $\nu >0$. A
(strong) {\em coreset} is a subset $C' \subseteq C$ of clients with associated weights
$\{w_c\}_{c \in C'}$ such that for any subset of facilities $S \subseteq F$ of size $|S|=k$
it holds that:
$$
(1-\nu) \cdot \sum_{c \in C} d(c,S)^z \le \sum_{c \in C'} w_c \cdot d(c,S)^z \le
(1+\nu)\cdot \sum_{c \in C} d(c,S)^z \,,
$$
\end{definition}
where, $z=1$ for \kmedian and $z=2$ for \kmeans.

\bigskip
To achieve this, we rely on the coreset construction by Cohen-Addad et al.~\cite{cohen2021new}.
\begin{theorem}[Cohen-Addad et al.~\cite{cohen2021new}, Theorem~1]
\label{theorem:coreset}
Given a metric instance $((U,d),C,F,k)$ of \kmedian or \kmeans. For each $\nu > 0$, there exists an algorithm that, with probability at least $1-\delta$, computes a coreset $C' \subseteq C$ of size $|C'| = \bigO(\Gamma\,k \log|U|)$ in time $\bigO(|U|)$ such that $\Gamma = \min(\nu^{-2} + \nu^{-z}, k \nu^{-2})$, where $z=1$ for \kmedian and $z=2$ for \kmeans.
\end{theorem}
To simplify subsequent discussions, we consider $\Gamma = 2 \cdot \nu^{-2}$ and regard the coreset size as $\bigO(\nu^{-2} k \log |U|)$.
Observe that the coresets obtained for \kmedian and \kmeans (using the above theorem) also serve as coresets for \kmedianppm and \kmeansppm, respectively. Furthermore, the coreset construction remains consistent for \divkmedian and \divkmeans, given that their respective clustering objective remains unchanged with respect to \kmedian and \kmeans.

\xhdr{Overview of $\fpt(\kt)$ algorithm}
A birds-eye-view of our \fpt algorithm is as follows: In Section~\ref{sec:algorithm:enumerating}, for a feasible instance of \divkmedian (\divkmeans resp.), we start by enumerating collections of facility subsets that satisfy lower and upper bound requirements. In Section~\ref{sec:algorithm:fpt-divkmedian}, for each collection satisfying the criteria, we derive a constant-factor approximation of the optimal cost with respect to that particular collection. Given that at least one of these feasible solutions contain the optimal solution, the corresponding approximate solution obtained will serve as an approximation for \divkmedian (\divkmeans resp.).

\subsection{Enumerating feasible constraint patterns} \label{sec:algorithm:enumerating}
We begin by defining the concept of a characteristic vector and a constraint pattern. 
Given an instance $\divkins$ of diversity-aware clustering with $\Gcal=\{G_i\}_{i \in [t]}$. A characteristic vector of a facility $f \in F$ is a vector denoted as $\charvec_f \in \{0,1\}^t$ such that $i$-th index of $\charvec_f$ is set to $1$ if $f \in G_i$, $0$ otherwise.
Consider the set $\{\charvec_f\}_{f\in F}$ of the characteristic vectors of facilities in $F$. For each $\gammavec \in \{0,1\}^t$, let $\patternset{\gammavec} = \{f \in F: {\charvec}_f=\gammavec\}$ denote the set of all facilities with characteristic vector $\gammavec$.  Finally, $\charpart=\{E(\gammavec)\}_{\gammavec \in \{0,1\}^t}$ induces a partition on $F$.
\begin{definition}[Constraint pattern]
Given a $k$-multiset $\Epcal=\{E(\gammavec_{i})\}_{i \in [k]}$, where each $E(\gammavec_{i}) \in \charpart$, the \pattern associated with $\Epcal$ is the vector obtained by the element-wise sum of the characteristic vectors $\{\gammavec_{1},\dots,\gammavec_{k}\}$, \ie, $\sum_{i \in [k]} \gammavec_{i}$.
A \pattern is said to be feasible if $\alphavec  \leq \sum_{i \in [k]} \gammavec_{i} \leq \betavec$, where the inequality is taken element-wise.
\end{definition}
We can enumerate all feasible constraint patterns for \divkmedian, \divkmeans, or \divksupplier in time $\bigO(2^{tk}t |U|)$, as asserted by the following lemma.

\begin{lemma}
\label{lemma:feasiblecp}
For a given instance of \divkmedian, \divkmeans, or \divksupplier, we can enumerate all $k$-multisets with feasible constraint pattern in time $\bigO(2^{tk}t |U|)$.
\end{lemma}
\begin{proof}
First, we construct the set $\charpart$ in time $\bigO(2^t|F|)$, since $|\charpart| \leq 2^t$. Then, we enumerate all possible $k$-multisets of $\charpart$ that have feasible constraint pattern. Since, there are $|\charpart|+k-1 \choose k$ possible $k$-multisets of $\charpart$, so enumerating all feasible {\pattern}s can be done in $\bigO(|\charpart|^k t)$ time, resulting in total running time $\bigO(2^{tk}t |U|)$.
\end{proof}
Facility groups with feasible constraint pattern $\Epcal=\{E(\gammavec)_i\}_{i \in [k]}$ might be duplicates. For instance, if it is necessary to choose two facilities from a group $E(\gammavec)$, we make a copy of facilities in $E(\gammavec)$ while introducing at most $\epsilon > 0$ distortion in distances, effectively forming a new group. Nonetheless, we consider every $\Epcal$ with feasible constraint pattern to be disjoint.

Observe that for every $k$-multiset $\Epcal=\{E(\gammavec_{i})\}_{i \in [k]}$ with a feasible \pattern, choosing any arbitrary facility from each $E(\gammavec_{i})$ yields a feasible solution to \reqsat.
More specifically, every feasible constraint pattern induces a \kmediankpm instance (\kmeanskpm, \ksupplierkpm, resp.).

\subsection{\fpt approximation algorithm for \divkmedian and \divkmeans}
\label{sec:algorithm:fpt-divkmedian}
In this section we present one of our main results. First, we give an intuitive overview of our algorithm for \divkmedian and \divkmeans. As a warm-up, we describe a simple $(3+\epsilon)$-\fpt-approximation algorithm for \divkmedian ($9+\epsilon$-approximation for \divkmedian resp.). Further, show how to obtain a better guarantee, leveraging the recent \fpt-approximation techniques of \kmedian and \kmeans.

\xhdr{Intuition}
Given an instance $\divkins$ of \divkmedian, we partition the facility set $F$ into at most $\bigO(2^t)$ subsets $\charpart = \{E(\gammavec)\}_{ \gammavec \in \{0,1\}^t }$, such that each subset $E(\gammavec)$ corresponds to the facilities with characteristic vector $\gammavec \in \{0,1\}^t$. Then, using Lemma~\ref{lemma:feasiblecp}, we enumerate all $k$-multisets of $\charpart$ with feasible \pattern. For each such $k$-multiset $\Epcal=\{E(\gammavec_{i})\}_{i \in [k]}$, we generate an instance $J_\Epcal=((U,d), \{E(\gammavec_{i})\}_{i \in [k]},C',k)$ of \kmedianpm, resulting in at most $\bigO(2^{tk})$ instances.
Next, using Theorem~\ref{theorem:coreset}, we build a coreset $C'\subseteq C$ of clients. Finally, we obtain an approximate solution to each \kmedianpm instance by adapting the techniques from Cohen-Addad et al.~\cite{cohen2019tight}.
For \divkmeans, the process is similar, except we create \kmeanskpm instance using $k$-multiset $\Epcal = \{E(\gammavec_i)\}_{i \in [k]}$ of feasible \pattern.

Consider  $\Epcal^*=\{E(\gammavec^*_{i})\}_{i \in [k]}$, a $k$-multiset of $\charpart$, corresponding to the optimal (feasible) constraint pattern $(\gammavec^*_i)_{i \in [k]}$. Then note that, after duplicating facilities to make elements in $\Epcal^*$ disjoint, $((U,d),\{E(\vec{\gamma}^*_{1}),\cdots,E(\vec{\gamma}^*_{k})\},C)$ is an instance of \kmediankpm, since we want to select at most $1$ facilities from each partition. However, since there are $k$ partitions, every feasible solution to \kmediankpm, must select exactly $1$ facility from each partition, as desired. 
In essence, using the \pattern, we have reduced \divkmedian to \kmediankpm.
Therefore, any approximation algorithm for \kmediankpm yields the same approximation factor for \divkmedian.

Let $\Epcal=\{E(\gammavec_{i})\}_{i \in [k]}$ be a $k$-multiset of $\charpart$ corresponding to a feasible constraint pattern $(\gammavec_i)_{i \in [k]}$, and let $J_\Epcal$ be the corresponding feasible \kmedianpm instance.
Let $\Ftilde^* = \{\ftilde_i^* \in E(\gammavec_{i})\}_{i \in [k]}$ be an optimal solution of $J_\Epcal$. For each $\ftilde^*_i$, let $\ctilde^*_i \in C'$ be a closest client, with $d(\ftilde^*_i,\ctilde^*_i) = \lambdatilde^*_i$. Next, for each $\ctilde^*_i$ and $\lambdatilde^*_i$, let $\Pitilde^*_i \subseteq E(\gammavec_{i})$ be the set of facilities $f\in E(\gammavec_{i})$ such that $d(f,\ctilde^*_i) = \lambdatilde^*_i$. Let us call $\ctilde^*_i$ and $\lambdatilde^*_i$ as the leader and radius of $\Pitilde^*_i$, respectively.
Observe that, for each $i \in [k]$,  $\Pitilde^*_i$ contains  $\ftilde^*_i$. Thus, if only we knew $\ctilde^*_i$ and $\lambdatilde^*_i$ for all $i \in [k]$, we would be able to obtain a provably good solution. See Figure~\ref{fig:mainfptapx} for an illustration.

To find the closest client $\ctilde^*_i$ and its corresponding distance $\lambdatilde^*_i$ in \fpt time, we employ techniques of Cohen-Addad et al.~\cite{cohen2019tight}, which they build on the work of Feldman and Langberg~\cite{feldman2011unified}. In our case, we use coreset construction by Cohen-Addad et al.~\cite{cohen2021new}.
The idea is to reduce the search spaces enough to allow brute-force search in \fpt time. 
First, note that we already have a client coreset $C'$ with $|C'| = \bigO(k\nu^{-2}\log |U|)$ from Theorem~\ref{theorem:coreset}. To find $\{\ctilde^*_i\}_{i \in [k]}$ we can enumerate all ordered $k$-multisets of $C'$ in $(\bigO(k\nu^{-2}\log |U|)^k$ time.
Then, to bound the search space of $\lambda^*_i$ (which is at most $\Delta = \poly(|U|)$), we discretize the interval $[1,\Delta]$ to $[[\Delta]_\eta]$, for some $\eta >0$. Note that $[\Delta]_\eta \le \lceil \log_{1+\eta} \Delta\rceil = \bigO(\log |U|/\eta)$. We thus enumerate all ordered $k$-multisets of $[[\Delta]_\eta]$ in at most $(\bigO(\log |U|/\eta))^k$ time. The total time to find $\ctilde^*_i$ and $\lambdatilde^*_i$ is thus $(\bigO(k^3\nu^{-3}\log^2 |U|))^{k}$, which is $\fpt(k)$.

Next, using the facilities in $\{\Pitilde^*_i\}_{i \in[k]}$, we find an approximate solution for the instance $J_\Epcal$. As a warm-up, we show in Lemma~\ref{lemma:threeapx} that picking exactly one facility from each $\Pitilde^*_i$ arbitrarily already gives a $(3+\epsilon)$ approximate solution. This fact follows, since $(U,d)$ is a metric space and satisfy triangle inequality.
In Lemma~\ref{lemma:partition} we obtain a better approximation ratio by modeling \kmedianpm as a problem of maximizing a monotone submodular function, relying on the ideas of Cohen-Addad et al.~\cite{cohen2019tight}. 

\begin{figure}
\centering
\scalebox{0.8}{\begin{tikzpicture}[scale=\tikzscale,every node/.style={scale=\tikzscale}]

\input{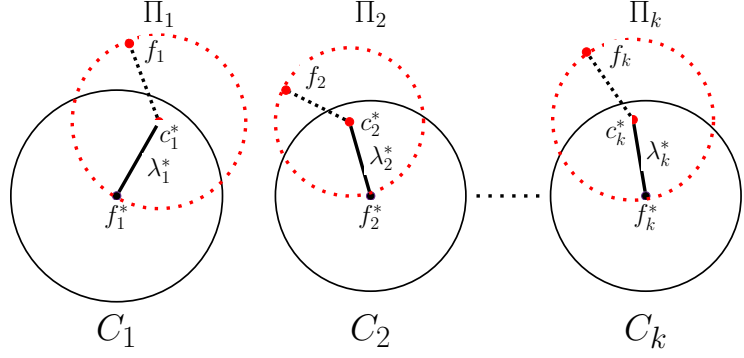}

\node[fcircle, minimum width=5cm] (c1) at (0,0) {};
\node[point1] (f1o) at (0,0) {};
\node[fill=white] at (0,-0.5) {\LARGE \bf $f_1^*$};
\node[fill=white] at (0,-3.2) {\Huge \bf $C_1$};

\node[fill=white] at (1,4.3) {\huge \bf $\smallball_1$};
\node[dcircle, minimum width=4.1cm] (c1) at (1,1.75) {};
\node[point2] (c1) at (1,1.75) {};
\node[fill=white] at (1.25,1.4) {\LARGE \bf $c_1^*$};
\node[point2] (f1) at (0.3,3.6) {};
\node[fill=white] at (0.9,3.4) {\LARGE \bf $f_1$};

\draw (f1o) edge[exedge] (c1);
\node[fill=white] at (1,0.6) {\LARGE \bf $\lambda_1^*$};
\draw (c1) edge[exedge, dotted] (f1);
\node[fcircle, minimum width=4.5cm] (c2) at (6,0) {};
\node[point1] (f2o) at (6,0) {};
\node[fill=white] at (6,-0.5) {\LARGE \bf $f_2^*$};
\node[fill=white] at (6,-3.2) {\Huge \bf $C_2$};

\node[fill=white] at (6,4.3) {\huge \bf $\smallball_2$};
\node[dcircle, minimum width=3.5cm] (c1) at (5.5,1.75) {};
\node[point2] (c2) at (5.5,1.75) {};
\node[fill=white] at (6,1.7) {\LARGE \bf $c_2^*$};
\node[point2] (f2) at (4.0,2.5) {};
\node[fill=white] at (4.7,2.8) {\LARGE \bf $f_2$};

\draw (f2o) edge[exedge] (c2);
\node[fill=white] at (6.25,0.9) {\LARGE \bf $\lambda_2^*$};
\draw (c2) edge[exedge, dotted] (f2);

\node[fcircle, minimum width=4.5cm] (c2) at (12.5,0) {};
\node[point1] (fko) at (12.5,0) {};
\node[fill=white] at (12.5,-0.5) {\LARGE \bf $f_k^*$};
\node[fill=white] at (12.5,-3.2) {\Huge \bf $C_k$};

\node[fill=white] at (12.5,4.3) {\huge \bf $\smallball_k$};
\node[dcircle, minimum width=3.8cm] (c1) at (12.2,1.8) {};
\node[point2] (ck) at (12.2,1.8) {};
\node[fill=white] at (11.8,1.5) {\LARGE \bf $c_k^*$};
\node[point2] (fk) at (11.1,3.4) {};
\node[fill=white] at (11.9,3.3) {\LARGE \bf $f_k$};

\draw (fko) edge[exedge] (ck);
\node[fill=white] at (12.8,1.0) {\LARGE \bf $\lambda_k^*$};
\draw (ck) edge[exedge, dotted] (fk);

\draw[loosely dotted, ultra thick, draw=black] (8.5,0) -- (10,0);
\end{tikzpicture}%
}
\caption{An illustration of facility selection for the \fpt algorithm for solving
\kmedianpm instance.
}
\label{fig:mainfptapx}
\end{figure}

\begin{lemma} 
\label{lemma:threeapx}
There exists a randomized $(3+\epsilon)$-approximation algorithm for \divkmedian, for every $\epsilon>0$,  which runs in time $f(k,t,\epsilon) \cdot \poly(|U|)$, where $f(k,t,\epsilon)=(\bigO(2^t \epsilon^{-3} k^3 \log^2 k))^k$.
\end{lemma}
\begin{proof} Let $I=\divkins$ be an instance of \divkmedian.
Let  $J=((U,d),C, \{E_1^*,\dots,E_k^*\},k)$ be an instance of \kmedianpm corresponding
to an optimal solution of $I$. That is, for some optimal solution $F^*=\{f_1^*, \dots,
f_k^*\}$ of $I$, we have $f_j^*\in E_j^*$. Let $c^*_j \in C'$ be the closest client to $f^*_j$, for $j \in [k]$, with $d(f^*_j,c^*_j) = \lambda_j$.
Now, consider the enumeration iteration where leader set is $\{c^*_j\}_{j \in [k]}$ and the radii is $\{\lambda_j^*\}$. The construction is illustrated in Figure~\ref{fig:mainfptapx}.
 
We define $\Pi^*_i$ to be the set of facilities in $E(\gammavec_i^*)$
at a distance of at most $\lambda_i^*$ from $c_i^*$. 
We will now argue that picking one arbitrary facility from each $\Pi^*_i$ gives a
$3$-approximation with respect to an optimal pick.
%
Let $C^*_j \subseteq C'$ be a set of clients assigned to each facility $f^*_j$
in optimal solution. Let $\{f_1,\dots,f_k\}$ be the  arbitrarily chosen
facilities, such that $f_j \in \Pi^*_j$. Then for any $c \in C_j$
\[
d(c,f_j) \leq d(c,f_j^*) + d(f_j^*,c_j^*) + d(c_j^*,f_j)).
\]

By the choice of $c_j^*$ we have 
$d(f_j^*,c_j^*) + d(c_j^*,f_j)) \leq 2 \lambda_j^* \leq 2 d(c,f_j^*)$, 
which implies
$\sum_{c \in C_j} d(c,f_j) \leq 3 \sum_{c \in C_j} d(c,f_j^*).$
By the properties of the coreset and bounded discretization
error~\cite{cohen2019tight}, we obtain the approximation stated in the lemma.
\end{proof}
Similarly for \divkmeans, by considering squared distances, we obtain a $(9+\epsilon)$-approximation algorithm with same time complexity.
\begin{corollary}
For every $\epsilon>0$, there exists a randomized $(9+\epsilon)$-approximation algorithm for \divkmeans which runs in time $f(k,t,\epsilon) \cdot \poly(|U|)$, where $f(k,t,\epsilon)=(\bigO(2^t \epsilon^{-3} k^3 \log^2 k))^k$.    
\end{corollary}

\bigskip\noindent
Next, we focus on our main result, stated in Theorem~\ref{theorem:mainfptapx}. As mentioned before, we build upon the ideas for \kmedian of Cohen-Addad et al. of~\cite{cohen2019tight}. Their algorithm, however, does not apply directly to our setting, as we have to ensure that the chosen facilities should satisfy the constraints in $\alphavec$ and $\betavec$.
A key observation is that by relying on the partition-matroid constraint of the auxiliary submodular optimization problem, we can ensure that the output solution will satisfy the \pattern. Since at least one \pattern contains an optimal solution, we obtain the advertised approximation factor. The pseudocode of algorithms for \divkmedian and \kmediankpm is available in Algorithm~\ref{algo:divkmed} and \ref{algo:kmedpm}, respectively.

In the following lemma, we argue that this is indeed the case. Next, we will provide an analysis of the running time of the algorithm. This will complete the proof of Theorem~\ref{theorem:mainfptapx}.
\begin{theorem}[Algorithm for \divkmedian/\divkmeans]
\label{theorem:mainfptapx}
For every $\epsilon>0$, there exists a randomized $(1 + \frac{2}{e}+\epsilon)$-approximation algorithm for \divkmedian with running time $f(k,t,\epsilon) \cdot \poly(|U|,k,t)$, where $f(k,t,\epsilon) = \left(\bigO\left(\frac{k^3 \log^{2} k}{\epsilon^2 \cdot \log(1+\epsilon)}\right)\right)^k$.
For \divkmeans, with the same running time, we obtain a $(1 + \frac{8}{e} +\epsilon)$-approximation algorithm.
Assuming \GAPETH, the approximation ratios are tight for any $\fpt(\kt)$-algorithm.
\end{theorem}
\begin{proof}
The hardness results follows from~\cite[Theorem~2]{cohen2019tight}, which rules out factor $(1 + \frac{2}{e}-\epsilon)$-approximation algorithm for \kmedian and $(1 + \frac{8}{e} -\epsilon)$-approximation algorithm for \kmeans running in $\fpt(k)$ time, assuming \GAPETH. Therefore, any approximation algorithm achieving aforementioned factors for \divkmedian and \divkmeans running in time $\fpt(\kt)$ also yields approximation algorithm for \kmedian and \kmeans with the same factors running in time $\fpt(k+1) = \fpt(k)$, since $t=1$ for these problems.\footnote{In fact, \cite[Theorem~2]{cohen2019tight} rules out much stronger running times. Further, using the same assumption, Manurangsi~\cite{manurangsi2020tight} strengthens the running time lower bound and rules out $g(k)|U|^{o(k)}$ time, for any function $g$, for \kmedian and \kmeans. This implies no $g(k,t)|U|^{o(k)}$ time approximation algorithm can achieve better factors than that of Theorem~\ref{theorem:mainfptapx} for \divkmedian and \divkmeans. Note that, these problems can be exactly solved in time $|U|^{k+O(1)}$ time.}

Our algorithm is described in detail in Algorithm~\ref{algo:divkmed}. We primarily focus on \divkmedian, indicating the parts of the proof for \divkmeans. In essence, to achieve the results for \divkmeans, we need to consider squared distances which results in the claimed approximation ratio with same runtime bounds. The proof of Lemma~\ref{lemma:partition} will complete the proof of Theorem~\ref{theorem:mainfptapx}.

As mentioned before, to get a better approximation factor, the idea is to reduce the problem of finding an optimal solution to \kmediankpm to the problem of maximizing a monotone submodular function. 
To this end, for each $S \subseteq F$, we define the submodular function $\impr(S)$ that, in a way, captures the cost of selecting $S$ as our solution. To define the function $\impr$, we add a fictitious facility $f'_i$, for each $i \in [k]$ such that $f'_i$ is at a distance $2\lambda^*_i$ for each facility in $\Pi_i$. We, then, use the triangle inequality to compute the distance of $f'_i$ to all other nodes. Then, using an $(1-1/e)$-approximation algorithm (Line~12), we approximate $\impr$. Finally, we return the set that has the minimum \kmedian cost over all iterations.

\spara{Running time analysis.} First we bound the running time of Algorithm~\ref{algo:kmedpm}. Note that, the runtime of  Algorithm~\ref{algo:kmedpm} is dominated by the two \textit{foreach} loops (Line~1 and 6), since the remaining steps, including finding an approximate solution to the submodular function \impr, run in time $\poly(|U|)$. The \textit{for} loop of clients (Line~2) takes time $ (\bigO(k\nu^{-2} \log |U|))^k$. Similarly, the \textit{for} loop of discretized distances (Line~3) takes time $([\Delta]_\eta)^k = (\bigO(\log_{1+\eta} |U|))^k$, since $\Delta = \poly(|U|)$. Hence, setting $\eta = \Theta(\epsilon)$, the overall running time of Algorithm~\ref{algo:kmedpm} is bounded by\footnote{We use the fact that, if $k \le \frac{\log |U|}{\log\log |U|}$, then $(k \log^2 |U|)^k = \bigO(k^k \poly(n))$, otherwise if $k \ge \frac{\log |U|}{\log \log |U|}$, then $(k \log^2 n)^k = \bigO(k^k (k \log k)^{2k})$.}
\[
\left(\bigO\left(\frac{k \log^{2} |U|}{\epsilon^2 \cdot \log(1+\epsilon)}\right)\right)^k  \poly(|U|) = \left(\bigO\left(\frac{k^3 \log^{2} k}{\epsilon^2 \cdot \log(1+\epsilon)}\right)\right)^k  \poly(|U|) 
\]
Since Algorithm~\ref{algo:divkmed} invokes Algorithm~\ref{algo:kmedpm} $\bigO (2^{tk})$ times, its running time is bounded by
$
\bigO\left(\left( \frac{2^t k^3 \log^2 k}{\epsilon^2
	\log(1+\epsilon)}\right)^k \poly(|U|) \right)$.

\spara{Approximation guarantee.} We establish the approximation ratio in Lemma~\ref{lemma:partition}, which concludes the proof.
\end{proof}

\begin{algorithm}[h]
\caption{\sc Div-$k$-Median$(I=\divkins,\epsilon)$}
\label{algo:divkmed}
\KwIn{$I$, an instance of the \divkmedian problem\\
\Indp \Indp ~$\epsilon$, a real number}

\KwOut{$T^*$, subset of facilities}

\ForEach{$\gammavec \in \{0,1\}^t$} {
    $E({\gammavec}) \gets \{f \in {F} : \gammavec = \vec{\chi}_{f}\}$
}
$\Pcal \gets \{E(\gammavec): {\gammavec} \in \{0,1\}^t \}$

$C' \leftarrow \textsc{coreset}((U,d),\pazocal{F},C,k, \nu \leftarrow \epsilon/16)$

$T^* \leftarrow \emptyset$\\
\ForEach{multiset $\{E(\gammavec_1),\cdots,E(\gammavec_k)\} \subseteq \Pcal$ of size $k$} {
  \If{$\alphavec \leq \sum_{i \in [k]} \gammavec_{i} \leq \betavec$, element-wise} {
    Duplicate facilities to make subsets in $\{E({\gammavec}_{1}), \dots,
E(\gammavec_{k})\}$ disjoint\\
    $T \gets
\textsc{$k$-Median-$k$-PM}((U,d),\{E(\gammavec_{1}),\cdots,E(\gammavec_{k})\},C',\epsilon/4)$\\
    \If{$\textsf{cost}(C',T) < \textsf{cost}(C',T^*)$}{
      $T^* \gets T$
    }
  }
}
\Return{$T^*$}
\end{algorithm}
\begin{algorithm}
\caption{\sc ${\kmediankpm(J=((U,d),\{E_1,\cdots,E_k\},C'), \epsilon')}$}
\label{algo:kmedpm}
\KwIn{$J$, an instance of the \kmediankpm problem\\
\Indp \Indp ~$\epsilon'$, a real number}

\KwOut{$S^*$, a subset of facilities}

$S^* \leftarrow \emptyset$,
$\eta \leftarrow \frac{e \epsilon'}{2}$\\
\ForEach{ordered multiset $\{c'_{i_1},\cdots, c'_{i_k}\} \subseteq C'$ of size $k$} {
  \ForEach{ordered multiset $\Lambda = \{\lambda_{i_1},\cdots, \lambda_{i_k}\}$ such that $ \lambda_{i_1}
\subseteq [[\Delta]_\eta]$} {
    \For{$j=1$ to $k$} {
      $\Pi_j \leftarrow \{f \in E'_j \mid d_D(f,c'_{i_j}) = \lambda_{i_j} \}$ \label{algo:kmedpm:pi}\\
      Add a fictitious facility $F'_j$ \\
      \For{$f \in \Pi_j$} {
        $d(F'_j,f) \gets 2\lambda_{i_j}$
      }
      \For{$f \notin \Pi_j$} {
        $d(F'_j,v) \gets 2\lambda_{i_j} + \min_{f \in \Pi_j} d(f,v)$
      }
    }
    \For{$S \subseteq F$, define $\textsf{improve}(S) := \textsf{cost}(C',F') - \textsf{cost}(C',F' \cup
S)$} {
      $S_{max} \leftarrow S \subseteq  \pazocal{F}$ that maximizes
      $\textsf{improve}(S)$ s.t. $|S \cap \Pi_j|=1, \forall j\in[k]$\\
      \If{$\textsf{cost}(C',S_{max}) < \textsf{cost}(C',S^*)$} {
        $S^* \leftarrow S$
      }
    }
  }
}
\Return $S^*$
\end{algorithm}

\begin{lemma} 
\label{lemma:partition}
Let $I= \divkins$ be an instance of \divkmedian to Algorithm~\ref{algo:divkmed} and $F^*=\{f_i^*\}_{i \in [k]}$ being an optimal solution of $I$. Let $J=((U,d), \{E_i^*\}_{i \in [k]},C',k)$ be an instance of \kmedianpm corresponding to optimal solution $F^*$, \ie, $f_i^* \in E^*_i, i \in [k]$.
On input $(J,\epsilon')$, Algorithm~\ref{algo:kmedpm} outputs a set $\Shat$ satisfying $\cost(\Shat) \leq (1+\frac{2}{e}+\epsilon)\cost(F^*)$. 
Similarly, for \divkmeans, $\cost(\Shat) \leq (1+\frac{8}{e}+\epsilon)\cost(F^*)$.
\end{lemma}
\begin{proof}
Consider the iteration of Algorithm~\ref{algo:kmedpm} where the chosen clients and radii are optimal, that is, $\lambda_i^*=d(c_i^*,f_i^*)$ and this distance is minimal over all clients served by $f_i^*$ in the optimal solution.
Assuming the input described in the statement of the lemma, it is clear that in this iteration we have $f_i^* \in \Pi_i$ (see Algorithm~\ref{algo:kmedpm}, line~\ref{algo:kmedpm:pi}).
Furthermore, given the partition-matroid constraint imposed on it, the proposed
submodular optimization scheme is guaranteed to pick exactly one facility from
each of $\Pi_i$, for all $i$.

On the other hand, known results for submodular optimization show that this problem can be efficiently approximated within a factor $\left(1-1/e\right)$ of the optimum~\cite{calinescu2011maximizing}. This translates to a $(1+\frac{2}{e}+\epsilon)$-approximation ($1+\frac{8}{e}+\epsilon$ resp.) of the optimal choice of facilities, one from each of $\Pi_i$~\cite{cohen2019tight}. For complete calculations, see Appendix~\ref{app:fptapx}.
\end{proof}
Note that, Lemma~\ref{lemma:partition} results in $(1+\frac{2}{e}+\epsilon)$ and $(1 + \frac{8}{e}+\epsilon)$-approximation algorithms for \divkmedian and \divkmeans with disjoint facility groups, with running time $(\bigO( {k^3 \log^2 k}{\epsilon^{-2} \log(1+\epsilon)^{-1}}))^k \poly(|U|)$. Notably, it reduces the running time by a factor of $2^t$ as compared to corresponding problems with intersecting facility groups. 
Attaining this result involves transforming \divkmedian with disjoint facility groups into an instance of \kmediankpm, by duplicating the facility group $G_i \in \Gcal$  exactly $\alpha_i$ times and some additional prepossessing. This process effectively generates dummy facility groups. We remark that Algorithm~\ref{algo:kmedpm} can be extended to obtain similar results for \textsc{$k$-Median-$p$-Partition} and \textsc{$k$-Means-$p$-Partition}.

\begin{corollary}[Algorithm for \kmediankpm/\kmeanskpm] \label{cor:kmediankpm}
There exists a randomized $(1 + \frac{2}{e}+\epsilon)$-approximation algorithm for \kmediankpm, for any $\epsilon > 0$,  with running time $f(k,\epsilon) \cdot \poly(|U|,k)$, where $f(k,\epsilon) = \left(\bigO\left( \frac{k^2 \log^2 k}{\epsilon^2 \log(1+\epsilon)}\right) \right)^k$.
For \kmeanskpm, with the same running time a $(1 + \frac{8}{e} +\epsilon)$-approximation algorithm is achieved.
%
\end{corollary}

\begin{remark}
To extend our algorithmic  approach to fair $k$-median (fair $k$-means resp.) with (disjoint groups and) lower bound constraints, we duplicate the data points in group $G_i$ $\alpha_i$-times.\footnote{This will turn the metric into a psuedometric. Although, our algorithms are robust to this change, an alternative way to preserve the metricity is to slightly perturb the duplicate copies so that there is a non-zero but small appropriate distance between them.} This transformation creates an instance of \kmediankpm (\kmeanskpm resp.). Then, applying Corollary~\ref{cor:kmediankpm} yields a $\fpt(k)$-approximation algorithm with factors $1 + \frac{2}{e} + \epsilon$ for fair $k$-median and $1 + \frac{8}{e} + \epsilon$ for fair $k$-means with running time $f(k,\epsilon) \cdot \poly(|U|,k)$, where $f(k,\epsilon) = \left(\bigO\left( \frac{k^2 \log^2 k}{\epsilon^2 \log(1+\epsilon)}\right) \right)^k$. The approximation ratio is tight assuming \gapeth.
Extending these results to fair  $k$-median and fair  $k$-means with both lower and upper bound constraints remains an open problem.
\end{remark}

\subsection{\fpt approximation algorithm for \divksupplier} 
\label{sec:algorithm:fpt-divksupplier}
In our efforts to extend the coreset enumeration method described in Section~\ref{sec:algorithm:fpt-divkmedian} for addressing \divksupplier, we faced challenges. The primary obstacle arose from the lack of a known methodology for constructing coresets of an appropriate size for \ksupplier, say $f(k)\cdot \poly(\log n)$. Because of this limitation, we choose an alternative approach, circumventing coreset construction, and present a $5$-approximation algorithm for \divksupplier with time $\fpt(k+t)$.

An overview of our approximation algorithm for \divksupplier is as follows: we employ a reduction technique, akin to previous sections, to transform an instance $I$ of \divksupplier into $\bigO(2^{tk})$ instances of \ksupplierkpm in time $\bigO(2^{tk}t|U|)$. This reduction facilitates to identify instances whose feasible solutions satisfy the lower bound requirements of $I$. We guarantee that there is at least one feasible instance of \ksupplierkpm whose optimum cost is at most the optimum cost of $I$.
Next, we find a $5$-approximate solution in polynomial time for every instance of \ksupplierkpm using the $5$-approximation algorithm of Chen et al.~\cite{chen2024an} for \fairksupplier, which is a generalization of \ksupplierkpm. By combining two steps yields a $5$-approximation algorithm for \divksupplier with time $\fpt(k+t)$.

\begin{algorithm}[h]
\caption{\divksupplier $(I=\divkins)$}
\label{algo:divksupplier}
\KwIn{$I$, an instance of the \divksupplier problem}

\KwOut{$T^*$, subset of feasible facilities that is a $5$ approximation to $I$}

\ForEach{$\gammavec \in \{0,1\}^t$} {
    $E({\gammavec}) \gets \{f \in {F} : \gammavec = \vec{\chi}_{f}\}$
}
$\Pcal \gets \{E({\gammavec}): \gammavec \in \{0,1\}^t \}$


$T^* \leftarrow \emptyset$\\
\ForEach{multiset $\{E(\gammavec_{1}),\cdots,E(\gammavec_{k})\} \subseteq \Pcal$ of size $k$} {
  \If{$ \alphavec \leq \sum_{i \in [k]}\gammavec_{i} \geq \betavec $, element-wise} {
    Duplicate facilities to make subsets in $\{E(\gammavec_{1}), \dots,
E(\gammavec_{k})\}$ disjoint\\
    $T \gets 5$-approximation to  $((U,d),\{E(\gammavec_{1}),\cdots,E(\gammavec_{k})\},C)$ using the algorithm of Chen et al.~\cite{chen2024an} \\
    \If{$\textsf{cost}(C,T) < \textsf{cost}(C,T^*)$}{
      $T^* \gets T$
    }
  }
}
\Return{$T^*$}
\end{algorithm}

\begin{theorem}[Algorithm for \divksupplier]\label{thm:maindivksupplier}
There exists a $5$-approximation algorithm  for \divksupplier in time $2^{\bigO(tk)} \cdot \poly(|U|)$. Moreover, assuming $\fpt \ne W[2]$, there exists no $\fpt(k+t)$ algorithm that computes $(3-\epsilon)$-approximation, for any $\epsilon > 0$,  for \divksupplier.
\end{theorem}
\begin{proof}
The pseudocode of this algorithm is presented in Algorithm~\ref{algo:divksupplier}, whose proof is outlined as follows: mirroring the approach as detailed in Section~\ref{sec:algorithm:enumerating}, we reduce \divksupplier to $\bigO(2^{tk})$ instances of \ksupplierkpm. Let $I=\divkins$ be an instance of \divksupplier, and let $J=((U,d),F,\{E_1,\cdots,E_k\},C))$ denote an instance of \ksupplierkpm corresponding to an optimal solution of $I$.  Note that, we are required to select at most $1$ facility from each $E_i$ in $J$.
Leveraging the algorithm of Chen et al.~\cite{chen2024an}, we obtain a polynomial-time $5$-approximation for $J$. By selecting a solution that minimizes the objective over all $\bigO(2^{tk})$ instances, we can derive a $5$-approximation algorithm for \divksupplier. 

The hardness result follows due to Hochbaum et al.~\cite{hochbaum1986unified}, who showed that \ksupplier is \np-hard to approximate within a factor better than $3$ using a reduction from the hitting set ptoblem (\hittingset). Since \hittingset is $\wtwo$-hard, it implies that there is no $\fpt(k)$ algorithm that approximates \ksupplier better than factor $3$. Further, since \ksupplier is equivalent to \divksupplier when $t=1$, the claimed hardness result follows.\footnote{Assuming \GAPETH, for any function $g$, Goyal and Jaiswal~\cite{goyal2023tight} have shown that no approximation algorithm for \ksupplier with running time $g(k)|U|^{o(k)}$ achieves a factor better than $3$. This implies that, assuming \GAPETH, no approximation algorithm running in time $g(\kt)|U|^{o(k)}$ can achieve factor better than $3$ for \divksupplier. While \divksupplier can be solved exactly in time $|U|^{k+O(1)}$.}
\end{proof}


We summarize the computational complexity and algorithmic results in Table~\ref{table:summary}.

\begin{table}[t]
\footnotesize
\caption{\label{table:summary}A summary of computational complexity and
algorithmic results.
}
\centering
\begin{tabular}{l l l l l l l}
\toprule
\multirow{1}{*}{Objective} &
\multicolumn{3}{c}{Exact computation} &
\multicolumn{3}{c}{Approximation}\\
\cmidrule(l{0pt}r{5pt}){2-4} \cmidrule(l{5pt}r{5pt}){5-7}
 & $\poly(|U|,k,t)$ & $\fpt(k)$ & $\fpt(\kt)$ & $\poly(|U|,k,t)$ & $\fpt(k)$ & $\fpt(\kt)$\\
\midrule
& \multicolumn{6}{c}{Diversity-aware  clustering with lower and upper bound constraints} \\
\cmidrule(l{55pt}r{0pt}){1-7}
$k$-median & \np-hard & \wtwo-hard & \wtwo-hard & \np-hard & \wtwo-hard & $\approx 1 + 2 e^{-1}$\\
$k$-means & \np-hard & \wtwo-hard & \wtwo-hard & \np-hard & \wtwo-hard & $\approx 1 + 8 e^{-1}$\\
$k$-supplier & \np-hard & \wtwo-hard & \wtwo-hard & \np-hard & \wtwo-hard & $5$\\
\midrule
& \multicolumn{6}{c}{Fair clustering with lower bound constraints and disjoint groups} \\
\cmidrule(l{55pt}r{0pt}){1-7}
$k$-median & \np-hard~\cite{thejaswi2021diversity} & \wtwo-hard~\cite{thejaswi2021diversity} & \multicolumn{1}{c}{---} & $7.081$~\cite{thejaswi2021diversity} & $\approx 1 + 2 e ^{-1}$ & \multicolumn{1}{c}{---} \\
$k$-means  & \np-hard & \wtwo-hard & \multicolumn{1}{c}{---} & $64$~\cite{thejaswi2021diversity} & $\approx 1 + 8 e^{-1}$ & \multicolumn{1}{c}{---} \\
$k$-supplier & \np-hard & \wtwo-hard & \multicolumn{1}{c}{---} & $5$~\cite{chen2024approximation} & Open & \multicolumn{1}{c}{---} \\
\bottomrule
\end{tabular}
\end{table}
\section{Limitations, future work and open problems}
In this section we discuss limitations of our work, and propose directions for future work and open problems in diversity-aware fair clustering.

\xhdr{Limitations}
The {\fpt} algorithms presented for \divkmedian and \divkmeans obtain theoretically optimal approximation factors, matching the lower bounds. However, they may not scale to large datasets due to large exponential factors in their running time.
%
Furthermore, although, theoretically, coreset constructions should introduce minimal distortion in distances with a small $\gamma > 0$ factor, in practice, achieving smaller coresets often requires a larger $\gamma$, leading to larger approximation factors. Therefore, designing more practical algorithms for diversity-aware  clustering bypassing coreset construction and leader guessing framework remains an open challenge. We remark that for restricted metric spaces, such as {those with a} Euclidean metric, our approach is likely to achieve better approximation ratios and running times, since {the} Euclidean metric admits smaller coresets~\cite{bandyapadhyay2024coresets,huang2020coresets}.

\xhdr{Future directions}
In the quest for fair algorithmic solutions in unsupervised machine learning, problems like clustering are frequently revisited with various constraints to achieve different notions of fairness. While these solutions address immediate needs, their long-term impact and effectiveness remain uncertain. 
Thus, designing clustering algorithms that account for long-term fairness---potentially through online variants that consider fairness over time rather than immediately---remains an open challenge.
Although many algorithmic techniques address lower, upper, and exact requirements on cluster centers separately, techniques that solve one type of constraint often do not apply to others. Whether a reduction or simplification could enable these existing methods to handle both constraints---to solve diversity-aware (fair)  clustering---is still unknown. Resolving this would mark a significant advancement in designing fair algorithms for clustering problems.

\xhdr{Open problems}
For \fairkcenter and \fairksupplier (with disjoint groups and lower bound constraints), polynomial-time algorithms with approximation factors of $3$ and $5$ are known~\cite{jones2020fair,chen2024approximation}, while the known lower bound of approximation is $2$ and $3$, respectively. Bridging these gaps remain an important open problem. Extending this question, one could ask for an $\fpt(\kt)$ approximation for \divkcenter and \divksupplier with approximation factor $2$ and $3$, respectively, matching the lower bound of approximation.
\section{Conclusions}
\label{sec:conclustion} 

In this work, we investigate diversity-aware clustering problems where potential cluster centers are associated with multiple attributes, leading to intersecting groups. The goal is to select cluster centers to avoid both under and over representation by adhering to specified upper and lower bound thresholds for each group, while simultaneously minimizing the clustering objective---either $k$-median, $k$-means, or $k$-supplier. 
We analyze the computational complexity of these problems, establishing their \np-hardness and polynomial-time inapproximability. Furthermore, we prove that even allowing exponential running time in multiple parameters that {may be} small in practice, finding an optimal solution remains intractable.
%
Finally, we present parameterized approximation algorithms with ratios $1 + 2 e^{-1}$ and $1 + 8 e^{-1}$  for \divkmedian and \divkmeans, respectively, that essentially match the lower bound on the achievable approximation ratios. For \divksupplier we present a parameterized approximation algorithm with factor $5$.

\bigskip

\newpage
\bibliographystyle{ACM-Reference-Format}
\bibliography{paper}

\newpage
\appendix
\section{Parameterized complexity theory}\label{appendix:sec:parameterized}

In this section we define notation and terminology related to parameterized complexity. For terms and notations not defined here we refer the interested reader to check~\cite{downey2013fundamentals,cygan2015parameterized}.

\begin{definition}[Fixed-parameter tractability]
A problem $P$ specified by input $x$ and a parameter $k$ is {\em fixed-parameter tractable} (\fpt) if there exists an algorithm $A$ to solve every instance $(x,k) \in P$ with running time of the form $f(k) |x|^{\bigO(1)}$, where $f(k)$ is function depending solely on the parameter $k$ and $|x|^{\bigO(1)} = \mathrm{poly}(|x|)$ is a polynomial independent of the parameter $k$.\footnote{With
an exception that, $\log n$ is allowed as part of function $f(k)$ if $n \ll k$.}
A problem $P$ is {\em fixed-parameter intractable} otherwise if no algorithm with running time of the form $f(k)|x|^{\bigO(1)}$ exists to solve $P$.
\end{definition}

There exists a family of optimization problems for which no polynomial-time
approximation algorithms are known, a general belief is that these problems are
{\em inapproximable} in polynomial-time. However, we can try to design
approximation algorithms for these problem in \fpt time, which gives rise to
{\em fixed-parameter approximation} algorithms, which we formally define now.

\begin{definition}[{\bf Fixed-parameter approximation}]
An optimization problem $P$ where each input instance $x \in P$ is associated
with a corresponding optimal solution $s^*$ and parameters $\{k_1,\dots,k_q\}$
is {\em fixed-parameter approximable} with
respect to parameters $k_1,\dots,k_q$
if there exists an algorithm $\Acal$ that solves every instance 
$(x, k_1,\dots,k_q) \in P$ with running time of the form 
$f(k_1,\dots,k_q)|x|^{\bigO(1)}$ and returns a solution $s$ with cost 
at most $\Phi(s) \leq \alpha \cdot \Phi(s^*)$, $\alpha \geq 1$ for a minimization problem, and 
at least $\Phi(s) \geq \alpha \cdot \Phi(s^*)$, $\alpha \leq 1$ for a maximization problem.
Otherwise, $P$ is {\em fixed-parameter inapproximable} with respect to
parameters $k_1,\dots,k_q$.
\end{definition}

The complexity class \fpt consists of problems for which parameterized
algorithms are known. However, there are certain problems for which no known
parameterized algorithms exist for various natural parameters of the problem, leading to the
$W$-hierarchy. The $W$-hierarchy extends beyond \fpt and comprises a hierarchy
of complexity classes, \ie,
\[ \fpt \subseteq W[1] \subseteq W[2] \subseteq \dots\]

To establish that a problem belongs to a class in the $W$-hierarchy, it is
essential to show a parameterized reduction from a known problem in the
corresponding complexity class of the $W$-hierarchy. A parameterized reduction
is slightly different than a many-to-one reduction used to establish
\np-hardness and it is defined as follows.

\begin{definition}[{\bf Parameterized reduction}]
Let $P$ and $P'$ be two problems.
A {\em parameterized reduction} from $P$
to $P'$ is an algorithm $\mathcal{A}$ that transforms an instance $(x, k) \in P$ to an
instance $(x',k') \in P'$ such that:
($i$) $(x,k)$ is a {\sf yes} instance of $P$ if and only if $(x',k')$ is a {\sf yes} instance
of~$P'$;
($ii$) $k' \leq g(k)$ for some computable function $g$; and
($iii$) the running time of the transformation $\mathcal{A}$ is
$f(k)|x|^{\bigO(1)}$, for some computable function $f$.
Note that $f$ and $g$ need not be polynomial functions.
\end{definition}

For more details on parameterized complexity, we refer the reader to the book by
Downey and Fellows~\cite{downey2013fundamentals}, as well as the book by
Cygan et al.~\cite{cygan2015parameterized}.

\section{Proof of Theorem~\ref{theorem:mainfptapx}} \label{app:fptapx}

First, we present the proof of Lemma~\ref{lemma:partition} and continue to
Theorem~\ref{theorem:mainfptapx}. We primarily focus on
\divkmedian, indicating the parts of the proof for \divkmeans. In essence, to
achieve the results for \divkmeans, we need to consider squared distances which
results in the claimed approximation ratio with same runtime bounds.

As mentioned before, to get a better approximation factor, the idea is to reduce
the problem of finding an optimal solution to \kmedianpm to the problem of
maximizing a monotone submodular function. To this end, for each $S \subseteq
\pazocal{F}$, we define the submodular function $\impr(S)$ that, in a way,
captures the cost of selecting $S$ as our solution. To define the function
$\impr$, we add a fictitious facility $F'_i$, for each $i \in [k]$ such that
$F'_i$ is at a distance $2\lambda^*_i$ for each facility in $\Pi_i$. We, then,
use the triangle inequality to compute the distance of $F'_i$ to all other
nodes.
Then, using an $(1-1/e)$-approximation algorithm (Line~12), we approximate
$\impr$. Finally, we return the set that has the minimum \kmedian cost over all
iterations. 

\mpara{Correctness:}
Given $I=\divkins$ . Let $\mathcal{J} :=\{J_\ell\}$ be the instances of
\kmedianpm generated by Algorithm~\ref{algo:divkmed} at Line~9 (For simplicity,
we do not consider client coreset  $C'$ here). For correctness, we show that
$F^* = \{f^*_i\}_{i \in [k]} \subseteq \pazocal{F}$ is feasible to $I$ if and
only if there exists $J_{\ell^*} =  ((V,d),\{E^*_{1},\cdots, E^*_{k}\},C) \in
\mathcal{J}$ such that $F^*$ is feasible to $J_{\ell^*}$. This is sufficient,
since the objective function of both the problems is same, and hence returning
minimum over $\mathcal{J}$ of an optimal (approximate) solution obtains an
optimal (approximate resp.) solution to $I$.
We need the following proposition for the proof.

\begin{proposition} \label{cl:setfacvec}
For all $E \in \mathcal{E}$, $f \in E$, we have $\vec{\chi}_f = \vec{\chi}_E$.
\end{proposition}
\begin{proof}
Fix $E \in \mathcal{E}$ and $f \in E$. Since $E \in \mathcal{E}$, there exists
$\vec{\gamma} \in \{0,1\}^t$ such that $E_{\vec{\gamma}} = E$. But this means
$\vec{\chi}_{E} = \vec{\gamma}$. On the other hand, since $f \in
E_{\vec{\gamma}}$, we have that $\vec{\chi}_f = \vec{\gamma}$.
\end{proof}

Suppose $F^* = \{f^*_i\}_{i \in [k]} \subseteq \pazocal{F}$ is a feasible
solution for $I: \sum_{i \in [k]} \vec{\chi}_{f^*_i} \ge \vec{r}$.
Then, consider the instance of \kmedianpm,
$J_{\ell^*} = ((V,d),\{E^*_{1},\cdots,E^*_{k}\},C)$ 
with $E^*_i = E_{\vec{\chi}_{f^*_i}}$, for all $i \in [k]$.
Since, 
\[
\sum_{i \in [k]} \vec{\chi}_{E^*_i} = \sum_{i \in [k]} \vec{\chi}_{f^*_i} \ge \vec{r}
\]
we have that $J_{\ell^*} \in  \mathcal{J}$.
Further, $F^*$ is feasible to $J_{\ell^*}$ since $F^* \cap E^*_{i} = f^*_i$ for
all $i \in [k]$.
For the other direction, fix an instance $J_{\ell^*} =  ((V,d),\{E^*_{1},\cdots,
E^*_{k}\},C) \in \mathcal{J}$ and a feasible solution $F^* =\{f^*_i\}_{i \in
[k]}$ for $J_{\ell^*}$. From Claim~\ref{cl:setfacvec} and the feasiblity of
$F^*$, we have $\vec{\chi}_{f^*_i} = \vec{\chi}_{E^*_i}$. Hence,
\[
 \sum_{i \in [k]} \vec{\chi}_{f^*_i} = \sum_{i \in [k]} \vec{\chi}_{E^*_i}  \ge \vec{r}
\]
which implies $F^* =\{f^*_i\}_{i \in [k]}$ is a feasible solution to $I$. To
complete the proof, we need to show that the distance function defined in
Line~10 is a metric, and \impr function defined in Line~11 is a monotone
submodular function. Both these proofs are the same as that in
\cite{cohen2019tight}. 

\mpara{Approximation Factor:} 
For $I=\divkins$, let
$I'=((U,d),C',F,\mathcal{G}, \vec{r},k)$  be the instance with client coreset
$C'$. 
Let $F^* \subseteq \pazocal{F}$ be an optimal solution to $I$, and let
$\Tilde{F}^* \subseteq \pazocal{F}$ be an optimal solution to $I'$. Then, from
\coreset Lemma~\ref{sec:algorithm:coresets}, we have that
\[
(1-\nu) \cdot \textsf{cost}(F^*,C) \le \textsf{cost}(F^*,C') \le (1+\nu) \cdot \textsf{cost}(F^*,C).
\]
The following proposition, whose proof closely follows that
in~\cite{cohen2019tight}, bounds the approximation factor of
Algorithm~\ref{algo:divkmed}.

\begin{proposition} \label{lem:fptapxpm}
For $\epsilon' >0$, let $J_\ell=((V,d),\{E_1,\cdots,E_k\},C',\epsilon')$ be an input to Algorithm~\ref{algo:kmedpm}, and let $S^*_\ell$ be the set returned. Then,
\[
\textsf{cost}(C',S^*_\ell) \le (1+2/e + \epsilon') \cdot \textsf{OPT}(J_\ell),
\]
where $\textsf{OPT}(J_\ell)$ is the optimal cost of \kmedianpm on $J_\ell$. Similarly, for \divkmeans,
\[
\textsf{cost}(C',S^*_\ell) \le (1+8/e + \epsilon') \cdot \textsf{OPT}(J_\ell),
\]
\end{proposition}
This allows us to bound the approximation factor of Algorithm~\ref{algo:divkmed}.
\begin{align*}
\textsf{cost}(T^*, C') 
=   \min_{J_\ell \in \mathcal{J}}  \textsf{cost}(S^*_\ell,C') 
\le  (1+2/e+\epsilon') \cdot \textsf{cost}(\Tilde{F}^*,C') \\
\le (1+2/e +\epsilon') \cdot \textsf{cost}(F^*,C')
\le (1+2/e +\epsilon')(1+ \nu) \cdot \textsf{cost}(F^*,C).
\end{align*}
On the other hand, we have $\textsf{cost}(T^*,C) \le (1+2\nu) \cdot \textsf{cost}(T^*, C')$. Hence, using $\epsilon'=\epsilon/4$ and $\nu = \epsilon/16$, we have
\begin{align*}
\textsf{cost}(T^*,C) &\le (1+2/e +\epsilon')(1+ \nu) (1+2\nu) \cdot \textsf{cost}(F^*,C) \\
&\le (1+2/e+\epsilon)\cdot \textsf{cost}(F^*,C)
\end{align*}
for $\epsilon\le 1/2$. Analogous calculations holds for \divkmeans. This finishes the proof of Lemma~\ref{lemma:partition}.
Now, we prove Proposition~\ref{lem:fptapxpm}.

\begin{proof}
Let $F^*_\ell$ be an optimal solution to $J_\ell$. Then, since
$\textsf{cost}(C',F'\cup F^*_\ell) = \textsf{cost}(C',F^*_\ell)$, we have that
$F^*_\ell$ is a maximizer of the function $\impr(\cdot)$, defined at Line~11.
Hence due to submodular optimization, we have that
\[
\impr(S^*_\ell) \ge (1-1/e) \cdot \impr(F^*_\ell).
\]
Thus,
\begin{align*}
    \cost(C',S^*_\ell) &= \cost(C',F' \cup S^*_\ell) \\
    &= \cost(C',F') - \impr(S^*_\ell)\\
    &\le  \cost(C',F') - (1-1/e) \cdot \impr(F^*_\ell)\\
    &= \cost(C',F') - (1-1/e) \cdot \left( \cost(C',F') - \cost(C',F^*_\ell) \right) \\
    &= 1/e \cdot \cost(C',F') + (1-1/e) \cdot \cost(C',F^*_\ell)
\end{align*}
\end{proof}

The following proposition bounds $\cost(C',F')$ in terms of $\cost(C',F^*_\ell)$.
\begin{proposition} 
 $\cost(C',F') \le (3+2\eta) \cdot \cost(C',F^*_\ell)$.
\end{proposition}
Setting $\eta = \frac{e}{2} \epsilon'$, finishes the proof  for \divkmedian,
\begin{align*}
 \cost(C',S^*_\ell) &\le (3+ e \epsilon')/e \cdot \cost(C',F^*_\ell) + (1-1/e)
\cdot \cost(C',F^*_\ell)\\ &\le (1+2/e+\epsilon') \cdot \cost(C',F^*_\ell).
\end{align*}
For \divkmeans, setting $\eta = \frac{e}{16} \epsilon'$, finishes the proof of Proposition~\ref{lem:fptapxpm},
\begin{align*}
\cost(C',S^*_\ell) &\le (3+ 2e\epsilon'/16  )^2/e \cdot \cost(C',F^*_\ell) + (1-1/e)
 \cost(C',F^*_\ell)\\ &\le (1+8/e+\epsilon') \cdot \cost(C',F^*_\ell).
\end{align*}
for $\eta \le 1$.
%
\begin{proof}
To this end, it is sufficient to prove that for any client $c' \in C'$, it holds
that $d(c',F') \le (3+2\eta) d(c',F^*_\ell)$. Fix $c' \in C$, and let
$f^*_{\ell_j} \in F^*_\ell$ be the closest facility in $F^*_\ell$ with
$\lambda^*_j$ such that $\lambda^*_j = d(c'_{i_j},f^*_{\ell_j})$. Now,
\[
d(c',f^*_{\ell_j}) \ge d(c'_{i_j},f^*_{\ell_j}) \ge (1+\eta)^{([\lambda^*_j]_D -1)} \ge \frac{\lambda^*_j}{1+\eta}.
\]
Using, triangle inequality and the above equation, we have,
\begin{align*}
d(c',F') &\le d(c,F'_j) \le  d(c',F^*_{\ell_j}) + d(F^*_\ell,F') \le
d(c',f^*_{\ell_j}) + 2 \lambda^*_j \\&\le (3+2\eta)  d(c',f^*_{\ell_j}).
\end{align*}
\end{proof}

\end{document}